%% file: Main.tex
\documentclass[11pt,english]{article}
\usepackage{amsmath,amssymb,amsthm}
\usepackage{multicol}
\usepackage[margin=1in]{geometry}
\usepackage{graphicx,color}
\usepackage{enumitem}
\usepackage{fullpage}
\usepackage[noblocks]{authblk}
\usepackage{tcolorbox}
\usepackage{babel}
\usepackage{wrapfig}
\usepackage{MnSymbol}
\usepackage{mdframed}
\usepackage{algorithm}
\usepackage{algpseudocode}
\usepackage{soul}

\usepackage{tablefootnote}
\usepackage{thm-restate}
\usepackage{bbding}
\usepackage{pifont}
\usepackage{nicefrac}
\usepackage{mathrsfs}

\DeclareMathAlphabet{\mathdutchcal}{U}{dutchcal}{m}{n}

\definecolor{Darkblue}{rgb}{0,0,0.4}
\definecolor{Brown}{cmyk}{0,0.61,1.,0.60}
\definecolor{Purple}{cmyk}{0.45,0.86,0,0}
\definecolor{Darkgreen}{rgb}{0.133,0.543,0.133}

\usepackage[colorlinks,linkcolor=Darkblue,filecolor=blue,citecolor=blue,urlcolor=Darkblue,pagebackref]{hyperref}
\usepackage[nameinlink]{cleveref}

\usepackage[colorinlistoftodos,prependcaption,textsize=tiny]{todonotes}
\newif\ifdraft 
\draftfalse

\newtheorem{theorem}{Theorem}
\newtheorem{lemma}{Lemma}

\newtheorem{definition}{Definition}
\newtheorem{claim}{Claim}
\newtheorem{observation}{Observation}
\newtheorem{corollary}{Corollary}

\newtheorem{remark}{Remark}
\newtheorem*{goal}{Research Goal}

\let\int\undefined

\newcommand{\poly}{\mathrm{poly}}

\newcommand{\int}{\mathsf{Int}}

\newcommand{\bunch}{\mathsf{Bunch}}
\newcommand{\defi}{\stackrel{\text{\tiny{def.}}}{=}}

\newcommand{\lp}{\mathcal{LP}}
\newcommand{\wlp}{\widehat{\mathcal{LP}}}
\newcommand{\vlp}{\overrightarrow{\mathcal{LP}}}
\newcommand{\lvec}[1]{\overrightarrow{#1}}

\newcommand{\lspa}{\overrightarrow{\mathcal{SP}}}
\newcommand{\patt}{\mathbf{p}}
\newcommand{\ecc}{\mathsf{ecc}}
\newcommand{\tuple}{\mathsf{Tuple}}
\newcommand{\reach}{\mathsf{reach}}

 \newcommand{\eps}{\epsilon}

\DeclareMathOperator*{\argmax}{arg\,max}

\def\eps{\epsilon}

\newcommand{\initOneLiners}{%
	\setlength{\itemsep}{0pt}
	\setlength{\parsep }{0pt}
	\setlength{\topsep }{0pt}
}

\title{VC Set Systems in Minor-free (Di)Graphs and Applications}
\author{Hung Le}
\affil{University of Massachusetts at Amherst}
\author{Christian Wulff-Nilsen}
\affil{University of Copenhagen}

\date{}
\begin{document}
\maketitle
\begin{abstract}

A recent line of work on VC set systems in minor-free (undirected) graphs, starting from Li and Parter~\cite{LP19}, who constructed a new VC set system for planar graphs, has given surprising algorithmic results~\cite{LP19,Le22,DHV20,FMW20}. In this work, we initialize a more systematic study of VC set systems for minor-free graphs and their applications in both undirected graphs and directed graphs (a.k.a \emph{digraphs}).  More precisely:

\begin{enumerate}
	\item We propose a new variant of the Li-Parter set system for \emph{undirected} graphs. Our set system settles two weaknesses of the Li-Parter set system: the terminals can be anywhere, and the graph can be $K_h$-minor-free for any fixed $h$. We obtain several algorithmic applications, notably: (i) the first exact distance oracle for unweighted and undirected $K_h$-minor-free graphs that has truly subquadratic space and constant query time, and (ii) the first truly subquadratic time algorithm for computing Wiener index of $K_h$-minor-free graphs, resolving an open problem posed by Ducoffe, Habib, and Viennot~\cite{DHV20}.
	
	\item We extend our set system to   $K_h$-minor-free \emph{digraphs} and show that its VC dimension is $O(h^2)$.	We use this result to design the first subquadratic time algorithm for computing (unweighted) diameter and all-vertices eccentricities in   $K_h$-minor-free digraphs.
	
	\item   We show that the system of \emph{directed} balls in minor-free digraphs has VC dimension at most $h-1$. We then present a new technique to exploit the VC system of balls, giving the first exact distance oracle for unweighted minor-free digraphs that has truly subquadratic space and logarithmic query time.
	
	\item On the negative side, we show that VC set system constructed from shortest path trees of planar digraphs does not have a bounded VC dimension. This leaves an intriguing open problem: determine a necessary and sufficient condition for a set system derived from a minor-free graph to have a bounded VC dimension. 
\end{enumerate}

The highlight of our work is the results for digraphs, as we are not aware of known algorithmic work on constructing and exploiting VC set systems for digraphs.


\end{abstract}
\pagebreak
{\small \setcounter{tocdepth}{2} \tableofcontents}
\newpage
\pagenumbering{arabic}

\input{Intro}
\input{prelim}

\input{undirected}

\input{directed} 
\input{lowerbound}

\section{Conclusion}

In this work, we propose a systematic study of VC set systems in minor-free graphs, both directed and undirected. Our work leaves many open problems. First, could we establish a formal relationship between our set system $\wlp_{G,M}(S)$ and the original set system $\lp_{G,M}(S)$ by Li and Parter~\cite{LP19} in the sense that if one has a bounded VC dimension, then the other also does. This will imply that $\lp_{G,M}(S)$ is a VC set system for any $M,S$ and any minor-free graph $G$. The second question is to extend all results here to graphs beyond minor-free, such as graphs of polynomial expansion and nowhere-dense graphs. The third question is, could we design a truly subquadratic space distance oracle with \emph{constant} query time for minor-free digraphs? Our oracle in \Cref{cor1:oracle-digraph} has $O(\log n)$ query time. The fourth question is to obtain a similar metric compression result for digraphs. As far as we know, our \Cref{thm:LW-digraph} is not sufficient for metric compression as we do not have the triangle inequality in digraphs.

\paragraph{Acknowledgement.~} Hung Le is supported by the NSF CAREER Award No. CCF-2237288  and an NSF Grant No. CCF-2121952.
         \bibliographystyle{alpha}
	\bibliography{spanner}
	
	\pagebreak
	\appendix

\end{document}

%% file: Intro.tex
\section{Introduction}

A pair of seminar papers by Lipton and Tarjan~\cite{LT79,LT80} in the 70s initiated a productive line of research on planar graph algorithms. Over the past several decades, numerous algorithmic tools have been developed for planar graphs. We can roughly classify them into two classes: one for coping with NP-hard problems and another for designing fast algorithms for problems in P\footnote{We are referring to the optimization versions of decision problems in P and NP.}. The former class aims to provide (efficient) polynomial time approximation schemes or subexponential time (parameterized or exact) algorithms for NP-hard problems.   Representative examples are  Baker's layering technique~\cite{Baker94}, contraction decomposition~\cite{Klein05}, bidimensionality~\cite{DFHT5,DH05}, and sphere cut decomposition~\cite{DPBF9}, to name a few. The latter class aims to design (nearly) linear time, in many cases truly subquadratic time, algorithms for problems in P where no algorithms of the same running time were known for general graphs. A non-exhaustive list of examples includes the separator theorem~\cite{LT79,LT80} and $r$-division~\cite{Federickson87}, shortest path separator~\cite{LT79,Thorup04}, multiple-source shortest paths~\cite{Klein05B}, Voronoi diagram~\cite{Cabello18}, and VC-dimension~\cite{LP19}.  (The classification into two classes is not exclusive: there are techniques that can be used for both purposes.)

On the other hand, planarity is fragile: adding a single edge or vertex could make a planar graph become non-planar. Therefore, a major research goal is to extend the aforementioned algorithmic tools beyond planar graphs, specifically graphs that are more robust, such as bounded genus graphs and $K_h$-minor-free graphs. Bounded genus graphs are robust to edge addition---adding a new edge increases the genus by at most 1---but not to vertex addition as adding a single vertex could increase the genus by $\Omega(n)$. $K_h$-minor-free graphs are robust to both edge and vertex additions. Also, the class of $K_h$-minor-free graphs is vastly broader than the classes of planar and bounded genus graphs. 

Most algorithmic results mentioned above for planar graphs can be generalized to bounded genus graphs~\cite{Eppstein03,DHT04,CC07,DMM10} using now-standard topological tools. For minor-free graphs, the 20-year graph minor project by Robertson and Seymour provides a deep understanding of their structures~\cite{RS83,RS04}. The  Robertson-Seymour decomposition~\cite{RS03} has been used successfully to transfer almost all algorithmic tools in the first class (for coping with NP-hard problems) from planar graphs to $K_h$-minor-free graphs. However, the best-known algorithm for constructing the Robertson-Seymour decomposition has quadratic time~\cite{KKR12}, despite prolonged efforts to simplify the proofs of Robertson and Seymour~\cite{KTW18,KTW20}. The quadratic time makes the Robertson-Seymour decomposition inapplicable to transfer results from the second class to minor-free graphs. Furthermore, the dependency on the minor size $h$ is impractically huge even for a very small value of $h$. 
  As a result, there have been far fewer algorithmic tools for designing truly subquadratic time algorithms in $K_h$-minor-free graphs. Most focus has been on finding separators, and hence $r$-divisions, in minor-free graphs in truly subquadratic  time~\cite{WR09,KR10,WulffNilsen11,WulffNilsen14}. This deficiency motivates our work.
  
\begin{tcolorbox}[colback=white,bottom=0.2mm] 
	\begin{goal}\label{goal:main}	 \hypertarget{goal}{Enriching} the algorithmic toolkit for designing truly subquadratic time algorithms in $K_h$-minor-free graphs.
	\end{goal}
\end{tcolorbox}

Towards realizing our goal, we propose a systematic study of VC set systems (see \Cref{subsec:vc-def} for definitions) and their applications in designing truly subquadratic time algorithms. Our work was directly inspired by two recent results; both led to several surprising algorithmic applications.

The first is by Li and Parter~\cite{LP19}, who constructed a VC set system from a set of terminals lying on the outer face of a \emph{planar graph}. However, it remains unclear how to extend their results to $K_h$-minor-free graphs since the notion of the outer face is not well-defined, and their proof makes heavy use of planarity. The second is by Ducoffe, Habib, and Viennot~\cite{DHV20}, who designed the first truly subquadratic time algorithms for diameter and related problems in $K_h$-minor-free graphs via the VC set system of balls studied by Chepoi, Estellon, and Vaxes~\cite{CEV07}. However, the set system of balls is very difficult to work with algorithmically; this difficulty also manifests in the construction of Ducoffe, Habib, and Viennot~\cite{DHV20}, resulting in complicated algorithms. Consequently, the running time of their algorithms degrades exponentially in the size of the minor. 

 We remark that both results~\cite{LP19,DHV20} only apply to \emph{undirected graphs}, while our results extend to \emph{directed graphs} as well, which are often much harder to work with. Indeed, we are not aware of any VC set system for directed graphs, let alone using them in algorithmic applications. The pioneering work of Chepoi, Estellon, and Vaxes~\cite{CEV07} for planar graphs and of Kranakis et al.~\cite{KKRUW97} for general graphs do not consider directed graphs.

\subsection{VC Set Systems and Dimension}\label{subsec:vc-def}

A \emph{set system} is a pair $(U,\mathcal{F})$ where $U$ is a ground set and $\mathcal{F}$  is a collection of subsets of $U$; we only write $\mathcal{F}$ when the ground set is clear from the context. We say that $Y\subseteq U$ is \emph{shattered} by $\mathcal{F}$ if $\{Y\cap S: S \in \mathcal{F}\} = 2^{Y}$. That is, the intersections of $Y$ and the sets in $\mathcal{F}$ contain every subset of $Y$.   The \emph{VC-dimension} of a set system $(U,\mathcal{F})$ is the size of the largest subset $Y\subseteq U$ shattered by $\mathcal{F}$. The notion of VC-dimension was introduced by Vapnik and Chervonenkis~\cite{VC71}. We say that   $(U,\mathcal{F})$  is a \emph{VC set system} if its VC-dimension is bounded by a fixed constant.

Let $G$ be an edge-weighted and \emph{undirected} graphs. For a vertex $v\in V$ and a non-negative real number $r$, denote by $B(v,r) = \{u: d_G(u,v)\leq r\}$ a ball of radius $r$ centered at $v$.  Let $\mathcal{B}(G) = \{B(v,r): v\in V, r\in \mathbb{R}^+\}$ be the set of balls of all radii in $G$.  Chepoi, Estellon and Vaxes~\cite{CEV07} showed that $\mathcal{B}(G)$ has VC-dimension at most $4$ if $G$ is planar and remarked that the same proof should extend to any $K_h$-minor-free graphs; the proof then was given in detail by Bousquet and Thomass{\'{e}}\cite{BT15}.

\begin{theorem}[Chepoi, Estellon, and Vaxes~\cite{CEV07}]\label{thm:ball-vc} If $G$ is undirected and $K_h$-minor-free, then $(V,\mathcal{B}(G))$ has VC-dimension at most $h-1$. 
\end{theorem}

\Cref{thm:ball-vc} had been used exclusively in graph theory and combinatorics~\cite{CEV07,BC14,BT15} until very recently, Ducoffe, Habib, and Viennot~\cite{DHV20} exploited this result algorithmically. Specifically, they  designed the first algorithm for computing the exact diameter and its variants, of minor-free graphs in truly subquadratic time. They relied on a deep result of Haussler and Welzl~\cite{HW87}, who showed that any VC set system admits a spanning path with sublinear \emph{stabbing number}. They skillfully combined  the low-stabbing spanning path technique with the $r$-division technique, a standard tool in designing algorithms in minor-free graphs on which almost all truly subquadratic time algorithms rely. Indeed, they had to work very hard to fit both techniques together (Lemma 5.2 in \cite{DHV20}). However, there remain two undesirable aspects of their algorithm. 

First, it is difficult to adapt their algorithms to other problems. One specific problem is computing the Wiener index, i.e., the sum of all-pairs distances. They wrote, ``we currently do not see any way to extend our approach [...] to also compute their Wiener index  in truly subquadratic time." The Wiener index problem was rooted in chemistry~\cite{Wiener47} and has been studied extensively, e.g. see~\cite{CK97,CK09,WulffNilsen2009,Cabello18,GKMSW21}. As we will see in \Cref{subsec:Wiener}, the Wiener index problem can be readily handled by our technique.  Second, the final running time  degrades exponentially in $h$:  $O(n^{2-\varepsilon_h})$ where $\varepsilon_h = 2^{-O(h)}$. (The precise value of $\varepsilon_h$ is not given in~\cite{DHV20}.)

In a completely different context, motivated by the diameter problem in the distributed CONGEST model, Li and Parter~\cite{LP19} set up a different VC set system from a fixed set of terminals $S$. In their paper, they only studied a special case where $S$ contains vertices  \emph{on the outer face of a planar graph}, though the definition applies to any $S$.

\begin{definition}[Li-Parter~\cite{LP19}]\label{def:LP-setsystem} Let $M \subseteq \mathbb{R}$ be a set of real numbers. Let $S= \langle s_0,\ldots,s_{k-1}\rangle $ be a sequence of $k$ vertices  in an undirected and edge-weighted graph $G$. For every $v\in V$, define:
	\begin{equation}\label{eq:distance-index}
		X_v = \{(i,\Delta): 1\leq i \leq k-1, \Delta \in M, d_G(v,s_{i})-d_G(v,s_{i-1})\leq \Delta\} 
	\end{equation} 
	Let $\lp_{G,M}(S) = \{X_v: v\in V\}$ be a collection of subsets of the ground set $[k-1]\times M$. 
\end{definition}

The complicated-looking set $X_v$ intuitively encodes the (approximate) distance from $v$ to each vertex in $S$: the pair $(i,\Delta)\in X_v$ indicates that  $d_G(v,s_{i}) \leq d_G(v,s_{i-1}) + \Delta$. Thus, given $d_G(v,s_0)$ and all the pairs $(i,\Delta)$ in $X_v$, we can iteratively recover an upper bound on $d_G(v,s_i)$ for any $i\in [2,k]$.  Depending on the choice of $M$, we might recover the exact or approximate distance $d_G(v,s_{i})$. Li and Parter showed that $\lp_{G,M}(S)$ is a VC set system for a special setting of $G$ and $S$ (Theorem 3.7 in~\cite{LP19}): 

\begin{theorem}[Li-Parter~\cite{LP19}]\label{thm:LP19} Let $G$ be an edge-weighted, undirected, planar graph. Let $S$ be a set  of $k$ vertices  ordered clockwise on the outer face of $G$. For any $M\subseteq \mathbb{R}$, $([k-1]\times M,\lp_{G,M}(S))$ has VC-dimension at most $3$. 
\end{theorem}

As $\lp_{G,M}(S)$ is capable of encoding the graph distances directly into the set system, it is much easier to use than $\mathcal{B}(G)$ in algorithm design. Specifically, it was instrumental in solving several problems in planar graphs: metric compression and distributed approximate diameter computation~\cite{LP19}, exact distance oracles~\cite{FMW20}, and approximate distance oracles~\cite{Le22}, despite the restriction on $S$ and $G$. A natural open problem is: can we remove the restriction on $S$ and $G$?

\subsection{Our Results and Techniques}

We propose several set systems in $K_h$-minor-free graphs: variants of $\lp_{G,M}(S))$ in both undirected graphs and digraphs, the set system of balls for digraphs, and a set system induced by shortest paths in digraphs. (We refer to directed graphs as \emph{digraphs}.)  We obtain both negative and positive results for these systems. We hope for a ``unified'' view of existing VC set systems to reconcile their differences and guide the development of new ones. Two VC set systems $\mathcal{B}_G$ and $\lp_{G,M}(S)$ differ in three aspects: (i) the ease of application, (ii) the scope of application \textemdash one for minor-free while the other for planar graphs \textemdash and (iii) the proof techniques. In terms of proof techniques,   Chepoi, Estellon, and Vaxes~\cite{CEV07} construct a $K_5$-minor directly assuming (for contradiction) that there is a large set of vertices shattered by $\mathcal{B}_G$ in a planar graph $G$; the end result is an elegant proof that can be easily extended to $K_h$-minor-free graphs (as done by Bousquet and Thomass{\'{e}}~\cite{BT15})).  We call this proof technique \emph{minor-building proof}. The proof of Li and Parter exhaustively considers different crossing patterns of paths between the terminals and hence heavily relies on the assumption that $G$ is planar and $S$ on the outer face to make the number of crossing patterns manageable.  

The proofs of our positive results in this work are minor-building, though each VC set system needs its own twist in the proof. Our proofs inherit the simplicity and elegance of the minor-building technique, and are applicable to both undirected graphs and digraphs, as described in \Cref{subsec:Vc-dim-LPhat}, \Cref{subsec:VC-dim-vecLP} and \Cref{subsec:VC-dim-diballs}. The minor-building proof technique is also instructive in developing new set systems. Indeed, in an (unsuccessful) attempt to reprove the result by Li and Parter (\Cref{thm:LP19}) using the minor-building technique, we came up with a  VC set system slightly different from  $\lp_{G,M}(S))$, which retains all the aforementioned strengths of $\lp_{G,M}(S))$ while addressing its two weaknesses: $G$ can be any $K_h$-minor-free graph, and $S$ could be anywhere in the graph. 

\begin{definition}\label{def:LP-variant} Let $M \subseteq \mathbb{R}, G,$ and $S$ as in \Cref{def:LP-setsystem}. For every $v\in V$, define:
	\begin{equation}\label{eq:X-widehat}
		\widehat{X}_v = \{(i,\Delta): 1\leq i \leq k-1, \Delta \in M, d_G(v,s_{i})-\mbox{\hl{$d_G(v,s_{0})$}}\leq \Delta \} 
	\end{equation}
	Let $\wlp_{G,M}(S) = \{\widehat{X}_v: v\in V\}$ be a collection of subsets of the ground set $[k-1]\times M$. 
\end{definition} 

$\widehat{X}_v$ differs $X_v$ (Equation~\eqref{eq:distance-index}) in the highlighted term: it uses $d_G(v,s_{i})- d_G(v,s_{0})$ instead of $d_G(v,s_{i})- d_G(v,s_{i-1})$. The difference, while superficially small, is technically important for the minor-building proof technique; see \Cref{rm:LP} for a more formal discussion of why the minor-building proof technique fails for the set system $\lp_{G,M}(S)$. This leads to our first main result:

\begin{restatable}{theorem}{LWUndir}\label{thm:LW-undirected}Let $S$ be any set of vertices on an edge weighted, undirected $K_h$-minor-free graph $G$. Let $M\subseteq \mathbb{R}$ be any set of real numbers. Then $\wlp_{G,M}(S)$ has VC-dimension at most $h-1$.
\end{restatable}

Here we sketch key ideas of our proof. In the prior minor-building techniques for the set system of balls, a crucial step is to choose the shattering family of sets, which is the set of balls that shatters a set of vertices of size $h$. There could be many such choices, and choosing the right tie-breaking scheme for these balls is important:  Chepoi, Estellon, and Vaxes~\cite{CEV07} broke ties by the sum of distances to be minimum, while Bousquet and Thomass{\'{e}}\cite{BT15} did so by the radii of the balls. However, $\wlp_{G,M}(S)$ is very different from a set system of balls, and we have to choose a different tie-breaking scheme for the shattering family of sets. It turns out that by defining $\widehat{X}_v$ as in \Cref{def:LP-variant}, we could choose a tie-breaking scheme using the Isolation Lemma~\cite{VV86}. The Isolation Lemma has been used in breaking ties in different applications, e.g. see~\cite{VV86,Erickson10,MNNW18,CCE13,BP21}, and we expect that this lemma will be used more in future work involving the minor-building technique.

In all applications of the VC set system $\lp_{G,M}(S)$ in planar graphs that we are aware of, including those mentioned in~\cite{LP19,FMW20,Le22}, we can use  $\wlp_{G,M}(S)$ while obtaining the same, or sometimes stronger, guarantees. For example, we could derive a metric compression scheme with almost the same guarantees obtained by Li and Parter~\cite{LP19} but without the assumption that $S$ must be on the outer face and furthermore, $G$ could be any minor-free graphs; see \Cref{subsec:other-undirected}.

Beyond planar graphs, which is our \hyperlink{goal}{Research Goal} mentioned above, we construct a distance oracle (see \Cref{sec:prelim} for the definition) for \emph{unweighted} $K_h$-minor-free graphs with truly subquadratic space and \emph{constant query time}. This is the first oracle in $K_h$-minor-free graphs achieving truly  subquadratic space-query time product, though many such oracles were  known in planar graphs\footnote{It might be possible to extend some distance oracles with truly subquadratic space-query product from planar graphs to bounded genus graphs; however, we are not aware of any prior paper in this direction.} years ago~\cite{FR01,MS12,CDW17,GMWW18,CGMW19,LP20}. Furthermore, our oracle can also be constructed in truly subquadratic time. ($\tilde{O}$ notation hides a poly-logarithmic factor in $n$.)

\begin{corollary}\label{cor1:oracle-unweighted}Let $G = (V,E)$ be an unweighted $K_h$-minor-free graph. We can construct an exact distance oracle for $G$ with $\tilde{O}(n^{2-\frac{1}{3h-1}})$ space and $O(1)$ query time. The construction time of our oracle is $\tilde{O}(n^{2-\frac{1}{3h-1}})$.
\end{corollary}

Our oracle in \Cref{cor1:oracle-unweighted} is obtained by tailoring the construction of Fredslund-Hansen, Mozes, and  Wulff-Nilsen to $K_h$-minor-free graphs and applying \Cref{thm:LW-undirected}  to bound the number of distance patterns; we refer readers to \Cref{subsec:oracle-undirected} for more details.

Using \Cref{thm:LW-undirected}, we resolve an open problem left by Ducoffe, Habib, and Viennot~\cite{DHV20}: computing the  Wiener index in any $K_h$-minor-free graph in truly subquadratic time. We also improve the truly subquadratic time algorithm for computing all-vertices eccentricities and diameter in unweighted $K_h$-minor-free graphs by Ducoffe, Habib, and Viennot~\cite{DHV20} from  $n^{2-1/2^{O(h)}}$ to $\tilde{O}(n^{2-\frac{1}{3h-1}})$

\begin{corollary}\label{cor:dimater-unweighted} Let $G = (V,E)$ be an unweighted  $K_h$-minor-free graph. We can compute the eccentricities of all vertices, the diameter, and the Wiener index  of $G$  in $\tilde{O}(n^{2-\frac{1}{3h-1}})$ time.
\end{corollary}

We remark that a truly subquadratic running time of the form $2^{o(h)}n^{2-\eps}$ for any fixed constant $\eps > 0$ for computing diameter in unweighted $K_h$-minor-free graphs is unlikely due to a conditional lower bound by Abboud, Williams, and Wang~\cite{AWW16}, which holds even in a special case of graphs of treewidth at most $h$.

We now describe our results for \emph{digraphs}. Let $d_G(u\rightarrow v)$ denotes the distance \emph{from $u$ to $v$} in a digraph $G$. It might be that  $d_G(u\rightarrow v) \not=  d_G(v\rightarrow u)$. Analogous to \Cref{def:LP-variant}, we define a set system, denoted by $\vlp_{G,M}(S)$.

\begin{definition}\label{def:LP-directed} Let $M \subseteq \mathbb{R}$ and $S= \{s_0,s_2,\ldots, s_{k-1}\}$ be in \Cref{def:LP-setsystem}, but $G = (V,E)$ now is an edge-weighted digraph. For every $v\in V$, let:
	\begin{equation}\label{eq:X-vec}
		\lvec{X}_v = \{(i,\Delta): 1\leq i \leq k-1, \Delta \in M, d_G(v \rightarrow s_{i})-  d_G(v \rightarrow s_{0})\leq \Delta \} 
	\end{equation}
We define $\vlp_{G,M}(S) = \{\lvec{X}_v: v\in V\}$.
\end{definition}

Our second main result is to show that $\vlp_{G,M}(S)$ is a VC set system in $K_h$-minor-free digraphs. (A digraph is $K_h$-minor-free if its underlying undirected graph is $K_h$-minor-free.)

\begin{restatable}{theorem}{LWDir}\label{thm:LW-digraph}Let $S$ be any set of vertices on an edge weighted $K_h$-minor-free digraph $G$. Let $M\subseteq \mathbb{R}$ be any set of real numbers. Then $\vlp_{G,M}(S)$ has VC-dimension at most $h^2$.
\end{restatable}

The VC-dimension bound in \Cref{thm:LW-digraph} is quadratic instead of linear as in \Cref{thm:LW-undirected}. Our proof of \Cref{thm:LW-digraph} is also minor-building. However, the main difficulty in the directed case is that two directed shortest paths could intersect an arbitrary number of times (in different directions). In the undirected case, we rely on the fact that two shortest paths intersect at most once, as long as we choose a consistent tie-breaking scheme.  The fact that directed paths can intersect in a very complicated way makes the minor construction in digraphs more difficult, and we settle on a looser bound. To construct a minor, we group the vertices into $h$ groups, and loosely speaking, we show that how to choose directed paths between groups so that the paths are vertex disjoint.

We use \Cref{thm:LW-digraph} to design first truly subquadratic time algorithm for computing diameter and eccentricity for unweighted $K_h$-minor-free digraphs. Previously, truly subquadratic time algorithms for these problems were only known for planar digraphs~\cite{Cabello18,GKMSW21}. 

\begin{corollary}\label{cor:dimater-digraph} Let $G = (V,E)$ be an unweighted  $K_h$-minor-free digraph. We can compute the diameter and all-vertex eccentricities of $G$  in $\tilde{O}(n^{2-1/(3h^2+6)})$ time.
\end{corollary}

Designing the truly subquadratic time algorithm for computing diameter and all-vertex eccentricities of digraphs in \Cref{cor:dimater-digraph} is much more difficult than their undirected counterparts in \Cref{cor:dimater-unweighted}. The algorithm for undirected graphs is based on the notion of \emph{patterns}: each pattern is intuitively a vector  of distances from a vertex in the graph to the boundary of a subgraph; the formal definition is given in \Cref{eq:pattern-def}. Two nice properties of patterns in undirected graphs: (i) there is only a polynomial number of them, and (ii) the distance from a vertex $u$ to a vertex $v$ in a connected subgraph of $H$ can be defined in terms of the distance from the pattern of $u$ to  $v$. (We have not defined the notion of distance between a pattern and a vertex; for now it suffices to know that one could define such a notion.)  In digraphs, property (i) breaks down completely, and the reason is perhaps unsurprising: the triangle inequality does not hold in digraphs ---the asymmetric triangle inequality does not suffice. Instead, we introduce \emph{infinite patterns} where we allow entries with $\pm +\infty$ values. For infinite patterns, we are able to obtain property (i). However, property (ii) fails for infinite patterns. We resolve this by looking at all the distances from the pattern to all vertices of $H$ at once, and we are able to extract the maximum distance from these distances. Thus, we are still able to solve the diameter and all-vertices eccentricities problems in truly subquadratic time. Unfortunately, we are not able to compute the Wiener index in truly subquadratic time using infinite patterns, and we leave this as  an open problem for future work.

In undirected graphs, we can use VC dimension bound on $\wlp_{G,M}(S)$ to construct an exact distance oracle with truly subquadratic space and constant query time (\Cref{cor1:oracle-unweighted}). However, we are unable to use the  VC dimension bound on $\vlp_{G,M}(S)$ to obtain an analogous result for digraphs. This is because the notion of patterns does not work, and the infinite patterns we introduce are not useful in decoding distances. We work around the problem in our third main result. Specifically, let $\lvec{B}(v,r) = \{u: d_G(v\rightarrow u)\leq r\}$, and:
\begin{equation}\label{eq:diball}
	\lvec{\mathcal{B}}(G) = \{\lvec{B}(v,r) : v\in V\}
\end{equation}

\begin{restatable}{theorem}{DiBallVC}\label{thm:diball-vc} If $G$ is a $K_h$-minor-free digraph, then $\lvec{B}(G)$ has VC-dimension at most $h-1$. 
\end{restatable}
 
We then develop a new technique to exploit the VC set system of directed balls. Our technique fits naturally with the $r$-division of $K_h$-minor-free digraphs. Specifically for each cluster in the $r$-division, we look at all the restrictions of balls in the cluster; the balls are centered at vertices outside the cluster. We exploit \Cref{thm:diball-vc} in showing that there are only a polynomial number of different restrictions. Thus, we could keep all of them, along with side information, in a table. Our technique gives the first exact distance oracle for digraphs with  truly subquadratic space-query product.  We remark that it is unclear how to combine the low-stabbing spanning path technique by Ducoffe, Habib, and Viennot~\cite{DHV20} with $r$-division to construct an exact distance with the same guarantee (even in undirected graphs).

\begin{corollary}\label{cor1:oracle-digraph}Let $G = (V,E)$ be an unweighted $K_h$-minor-free digraph. We can construct an exact distance oracle for $G$ with $\tilde{O}(n^{2-\frac{1}{2(h-2)}})$ space and $O(\log(n))$ query time. 
\end{corollary}

We now turn to a negative result. We study set systems whose ground set is the set of edges in digraphs. While there could be many ways to define a set system of edges~\cite{KKRUW97}, the system of shortest path trees is of special interest to us: such a set system, if has bounded VC-dimension, could be used to compute the Wiener index in truly subquadratic time---resolving the problem we pose above---speed up exact diameter computation, construct exact distance oracles for digraphs with $O(1)$ query time, and potentially has many more applications. Unfortunately, we show that the set system does not have bounded VC dimension.  More formally, given a digraph $G = (V,E)$,  let $\tau_v$ be the shortest path tree rooted at $v$. In the construction of shortest path trees in $G$, ties are broken consistently. (If ties are not broken consistently, it is fairly easy to show that the set system of edges introduced below will not have bounded VC dimension.) We think of $\tau_v$ as a subset of the $E$, and define:
\begin{equation}\label{eq:edge-sys-def} 
 \lspa(G) = \{\tau_v: v\in V\}
\end{equation}

As our fourth main result, we show that the set system $(E, \lspa(G) )$ does not have bounded VC dimension even in unweighted planar digraphs.

\begin{restatable}{theorem}{LBEdge}\label{thm:diEdge-vc} For any constant integer $r \geq 1$, there exists an unweighted planar digraph $G = (V,E)$ and  a subset $X\subseteq E$ of size $r$ such that $X$ is shattered by $\lspa(G)$.
\end{restatable}

Lastly, we briefly mention two other directions which we do not explore in this paper as they are out of scope. The first direction is to explore the applications of our VC dimension results in solving graph-theoretic problems. There have been several works on applying the prior VC dimension results by Chepoi, Estellon, and Vaxes (\Cref{thm:ball-vc}), for example \cite{BT15,BC14,BCEJLMP21}, and by Li and Parter (\Cref{thm:LP19}), for example~\cite{JR23}, to understand structures of planar and minor-free graphs. We believe that our results will also be applicable in this direction. The second direction is to consider graphs beyond minor-free, such as graphs with polynomial expansion or nowhere dense graphs, as studied in the work by Ducoffe, Habib, and Viennot~\cite{DHV20}. As far as we can see,  our results could also be extensible to graphs with polynomial expansion and get algorithmic applications along the line of  Ducoffe, Habib, and Viennot~\cite{DHV20}. However, it seems to us that one has to work harder to be able to extend our results to nowhere dense graphs. 

%% file: prelim.tex
\section{Preliminaries}\label{sec:prelim}

We use graphs to refer to \emph{undirected graphs}, while directed graphs will be called \emph{digraphs}. We reserve $V$ and $E$ for the vertex set and edge set of $G$, respectively. For any other graph $H$, we denote it vertex set by $V(H)$ and edge set by $E(H)$. We denote by $\pi(u,v,G)$ a shortest path between $u$ and $v$ in a graph $G$. If $G$ is a digraph, then we denote by $\pi(u\rightarrow v,G)$ the directed shortest path from $u$ to $v$. If the graph is clear from the context, we simply denote the shortest paths by $\pi(u,v)$ and $\pi(u\rightarrow v)$, respectively. 

The \emph{eccentricity} of a vertex $u$, denoted by $\ecc(u)$ in a graph $G$ is $\ecc(u) = \max_{v\in V}d_G(u,v)$. The diameter of $G$ is the maximum eccentricity: $\max_{u\in V}\ecc(u)$.  The Wiener index of a graph $G$ is defined to be the sum of all pairwise distances: $\frac{1}{2}\sum_{u\in V}\sum_{v\in V} d_G(u,v)$. The Wiener index, eccentricity, and diameter of digraphs are defined similarly, with $d_G(u\rightarrow v)$ being used in place of  $d_G(u,v)$.

We say that a subgraph $H$ of $G$ is \emph{induced} if every edge in $G$ between two vertices in $H$ also appears in $H$. We will use the \emph{$r$-division} of minor-free graphs in our algorithms. A \emph{cluster} is a \emph{connected, induced subgraph} of $G$. Let $C$ be a cluster of $G$. We say that a vertex $v\in C$ is a \emph{boundary vertex} if $v$ is adjacent to a vertex $u\in V\setminus V(C)$. We use $\partial C$ to denote the set of all boundary vertices of $C$.  An  \emph{$r$-division} of $G$ is a collection $\mathcal{R}$ of clusters $G$ such that every cluster $R\in \mathcal{R}$ has at most $r$ vertices. 

Our definition of $r$-division is somewhat non-standard in the sense that we do not have $O(\sqrt{|V(R)|})$ bound on the number of boundary vertices of each cluster $R$. It is called  \emph{$r$-clustering} in the paper of Wulff-Nilsen~\cite{WulffNilsen11}. Here we still call it an $r$-division as most of the intuition in the use of $r$-clustering comes from $r$-division.

Wulff-Nilsen~\cite{WulffNilsen11} showed that one can construct an $r$-division of any $K_h$-minor-free graphs such that the total number of boundary vertices, counted with multiplicity, is small. We note that in our applications, it is important that each cluster $R\in \mathcal{R}$ is a connected subgraph of $G$.  

\begin{lemma}[Wulff-Nilsen, Lemma 2~\cite{WulffNilsen11}]\label{lm:r-division} Let $G$ be a $K_h$-minor-free graphs with $n$ vertices, and $r\in [Ch^2\log n, n]$ for a sufficiently large constant $C$. For any fixed constant $\eps > 0$,  we can construct  in time $O(n^{1+\eps}\sqrt{r})$  an $r$-division, say $\mathcal{R}$, of $G$ such that (a) $\sum_{R\in \mathcal{R}}|\partial R| = \tilde{O}(nh/\sqrt{r})$, and (b) every cluster $R\in \mathcal{R}$ has $|V(R)|\leq r$ and $|\partial R| = \tilde{O}(h\sqrt{r})$.  Furthermore, the number of clusters in $\mathcal{R}$ is at most $\tilde{O}(h n/\sqrt{r})$.
\end{lemma}

One could obtain an $r$-division with a number of clusters being $\tilde{O}(h n/r)$ with a larger running time. For us, the weaker bound in \Cref{lm:r-division} suffices. 

In many of our results, we will use the following well-known Sauer–Shelah Lemma, which gives a polynomial upper bound on the size of a VC set system.

\begin{lemma}[Sauer–Shelah Lemma]\label{lm:SS}  Let $\mathcal{F}$ be a family of subsets of a ground set with $n$ elements. If  VC-dimension of $\mathcal{F}$ is at most $k$, then $|\mathcal{F}| = O(n^{k})$.
\end{lemma}

A distance oracle for a graph $G$ is a compact data structure that given any two vertices $u$ and $v$, returns $d_G(u,v)$ quickly. The query time is the maximum time it takes to answer a query over all pairs of vertices. There is often a trade-off between the space of the oracle and the query time.

%% file: undirected.tex
\section{VC Dimension of Undirected Graphs and Applications}

\input{LPGen}

\subsection{Algorithmic Applications}\label{subsec:apps-undirected}

In this section, we explore algorithmic applications of \Cref{thm:LW-undirected}. \emph{Graphs in this section are \emph{unweighted}, and hence the distances are unweighted distances.} We will use the notion of \emph{patterns}, introduced by Fredslund-Hansen, Mozes, and Wulff-Nilsen~\cite{FMW20}, though our pattern is defined slightly differently. Specifically, our definition rests on the VC set system $\wlp_{G,M}$, while Fredslund-Hansen, Mozes, and Wulff-Nilsen relied on the VC set system $\lp_{G,M}$ by Li and Parter~\cite{LP19}.

Let $H$ be a connected, induced subgraph of $G$. Recall that $\partial H$ denotes the set of all boundary vertices of $H$. Fix an arbitrary \emph{sequence} $\sigma_H$ of vertices of $\partial H$, which is a linear order of  $\partial H$. We write $\sigma_H = \langle s_0,s_1,\ldots, s_{|\partial H|-1}\rangle$. For each vertex $v\in V$, we define a \emph{pattern} of $v$ w.r.t $\sigma_H$, denoted by $\patt_v$, to be a $|\partial H|$ dimensional vector where:
\begin{equation}\label{eq:pattern-def}
	\patt_v[i] = d_G(v,s_i)- d_G(v,s_0) \quad\text{for every $0\leq i \leq |\partial H|-1$}
\end{equation}

Note that $\patt_v[0] = 0$ by definition. We bound the number of all possible patterns w.r.t. $\sigma_H$. 

\begin{lemma}\label{lm:pattern-bound-undir} Let $H$ be a connected, induced subgraph of a $K_h$-minor-free graph $G$, and $\sigma_H$ be an arbitrary sequence of vertices in $\partial H$. Let $P = \{\patt_v: v\in V\}$ be the set of all patterns w.r.t. $\sigma_H$. Then $|P| = O((|\partial H|\cdot |V(H)|)^{h-1})$. 
\end{lemma}
\begin{proof} Since $H$ is connected, by the triangle inequality, $ -(|V(H)|-1) \leq d_G(v,s_i)- d_G(v,s_0) \leq |V(H)|-1$. Let $M = \{-(|V(H)|-1),\ldots, -1,0,1,\ldots, (|V(H)|-1) \}$ and $S$ be the set of all boundary vertices of $H$. Let $\bar{p}_v$ be a set obtained by \emph{flattening} $\patt_v$; that is, for each $i\in [1,|\partial H|-1]$, we add to the set $\bar{p}_v$ a pair $(i,\Delta)$ for every $\Delta \in M$ such that $\Delta\geq \patt_v[i]$.  Observe by definition of $\widehat{X}_v$ in \Cref{eq:X-widehat} that $\bar{p}_v = \widehat{X}_v$. Thus, there is a bijection between the set of patterns $P$ and  $\wlp_{G,M}$.
	
	By the Sauer–Shelah Lemma (\Cref{lm:SS}), we have  $|\wlp_{G,M}| = O((|S||M|)^{h-1}) = O((|\partial H|\cdot |V(H)|)^{h-1}) $ as claimed.
\end{proof}

Let $v$ be a vertex in $H$, and $\patt$ be a pattern (of some vertex $u$) w.r.t. $\sigma_H$. We define the distance between $v$ in $\patt$, denoted by $d(\patt,v)$, to be:
\begin{equation}\label{eq:dis-patt-v}
	d(\patt,v) = \min_{0\leq i\leq |\partial H|-1}\{d_G(v,s_i) + \patt[i]\}
\end{equation} 

The distance between a vertex and a pattern can be used to compute the distance between two vertices as shown by the following lemma, due to Fredslund-Hansen, Mozes, and Wulff-Nilsen~\cite{FMW20}. Since our definition of a distance between a pattern and a vertex in \Cref{eq:dis-patt-v} is slightly different from that of \cite{FMW20}, we include a proof for completeness.

\begin{lemma}[Fredslund-Hansen, Mozes, and Wulff-Nilsen, Lemma 7~\cite{FMW20}]\label{lm:dist-via-pattern} Let $u \in V\setminus V(H)$ be a vertex not in $H$, and $\patt_u$ be the pattern of $u$ w.r.t $\sigma_H$. Let $v$ be a vertex in $H$. Then:
	\begin{equation}
		d_G(u,v) = d_G(u,s_0) + d(\patt_u,v)
	\end{equation}
\end{lemma}
\begin{proof}
	Observe that for each boundary vertex $s_{i}$ for $0\leq i\leq  |\partial H|$, $d_G(u,s_{i}) = \patt_u[i] + d_G(u,s_0)$. Let $s_{\ell}$ be the boundary vertex in $\pi(u,v,G)\cap \partial H$; $s_{\ell}$ exists since $H$ is an induced subgraph, and $u\not\in V(H)$, $v\in V(H)$. Then: 
	\begin{equation*}
		\begin{split}
			d_G(u,v) &= d_G(u,s_\ell) + d_G(s_\ell,v) = \min_{0\leq i \leq  |\partial H|-1}\{d_G(u,s_{i}) + d_G(s_{i},v)\}\\
			&=  \min_{0\leq i \leq  |\partial H|-1}\{ d_G(u,s_0) + \patt_u[i] + d_G(s_{i},v)\}\\
			&= d_G(u,s_0) + \min_{0\leq i \leq  |\partial H|-1}\{  d_G(v,s_i) + \patt_u[i]\}\\ & = d_G(u,s_0) + d(\patt_u,v)~,
		\end{split}
	\end{equation*}
	as desired.
\end{proof}

In \Cref{subsec:diameter-unweighted} and \Cref{subsec:Wiener}, we present algorithms to compute the diameter, eccentricities and the Wiener index. Our algorithm builds on an earlier algorithm by Wulff-Nilsen~\cite{WulffNilsen2009}. Here we use \Cref{lm:pattern-bound-undir} to improve the running time to truly subquadratic time.  In \Cref{subsec:oracle-undirected}, we construct a distance oracle with truly subquadratic space and constant query time. The algorithm is almost the same as the algorithm by Fredslund-Hansen, Mozes, and Wulff-Nilsen~\cite{FMW20}, except that we will use \Cref{lm:pattern-bound-undir}. In \Cref{subsec:other-undirected}, we mention other algorithmic applications.

\subsubsection{Diameter and Eccentricities}\label{subsec:diameter-unweighted}

In this section, we show how to compute all-vertices eccentricities in truly subquadratic time as described in \Cref{cor:dimater-unweighted}. Computing the diameter trivially follows by finding the maximum eccentricity in $O(n)$ time. The algorithm has three steps:


\begin{itemize}
	\item \textbf{(Step 1).~} Construct an $r$-division $\mathcal{R}$ of $G$ for $r = n^{2/(3h-1)}$. For each cluster $R\in \mathcal{R}$, form a sequence of boundary vertices $\sigma_R$ in an arbitrary way. Then compute the set of patterns w.r.t $\sigma_R$: $P_R = \{u\in V: \patt_u\}$. We store $P_R$ in a table  $T^{(1)}_R$.
	
	\item \textbf{(Step 2).~} For each cluster $R\in \mathcal{R}$ and each pattern $\patt \in P_R$, find $v = \argmax_{v\in V(R)} d(v,\patt)$. That is, $v$ is the vertex that has the maximum distance to $\patt$ over all vertices in $V(R)$; we say that $v$ is the \emph{furthest vertex} from $\patt$. We then store the distance $d(\patt,v)$ in a table $T^{(2)}_R$  of $R$; the key to access $T^{(2)}_R$ is (the ID of) $\patt$.  
	
	\item \textbf{(Step 3).~} We now compute $\ecc(u)$ for each vertex $u\in V$. For each cluster $R\in \mathcal{R}$, we compute the distance from $u$ to the vertex $v\in R$ furthest from $u$, denoted by $\Delta(u,R)$, as follows. \begin{itemize}
	    \item If $u\not\in R$, let $\patt_u$ be the pattern of $u$ w.r.t $\sigma_R$ computed in (Step 1). Let $v$ be the furthest vertex from $\patt_u$, computed in (Step 2). Then we return $\Delta(u,R) = d_G(u,s_0) + d(\patt_u,v)$ where $s_0$ is the first vertex of $\sigma_R$. Finally, we compute $\ecc(u) = \max_{R\in \mathcal{R}}\Delta(u,R)$. 
            \item   If $u\in R$, then we compute a distance $d_{R}(u,v)$ using BFS. Then, compute $\tilde{d}_G(u,v) = \min\{d_G(u,s_0) + d(\patt_u,v), d_{R}(u,v)\}$ and finally return $\Delta(u,R) = \max_{v\in R}{\tilde{d}_G(u,v)}$.
	\end{itemize}

\end{itemize}

By \Cref{lm:dist-via-pattern} and the computation in (Step 2),   if $u\not\in R$, then  $\pi(u,v,G)\cap \partial R\not=\emptyset$ and hence $\Delta(u,R)$ is correctly computed in Step 3. If $u\in R$, it is possible that  $\pi(u,v,G)$ contains outside, and in this case $d_G(u,v) = d_G(u,s_0) + d(\patt_u,v)$; otherwise, $d_G(u,v)  = d_R(u,v)$. As the algorithm takes the minimum, it correctly returns  $\Delta(u,R) = \max_{v\in R}{\tilde{d}_G(u,v)}$, and therefore, $\ecc(u)$ is correctly computed.

We now implement each step of the algorithm efficiently, assuming that $h$ is a constant. We can assume $h\geq 4$, as (connected) $K_3$-minor-free graphs are trees and hence all problems mentioned here can be solved in linear time. Let  $B$ be the set of boundary vertices of the $r$-division $\mathcal{R}$: $B = \cup_{R\in \mathcal{R}} \partial R$.  By \Cref{lm:r-division}, $|B| = \tilde{O}(n/\sqrt{r})$. Thus, we can find all BFS trees, each rooted at a vertex of $B$, in $\tilde{O}(n^2/r)$ time. 

\begin{observation}\label{obs:Tcal-Time} Let $D(B,V) = \{d_G(b,v): (b,v)\in B\times V \}$. Then $D(B,V)$ can be computed in time $\tilde{O}(n^{2}/\sqrt{r})$.
\end{observation}

By \Cref{lm:r-division}, each cluster $R\in \mathcal{R}$ has at most $r$ vertices and $\tilde{O}(\sqrt{r})$ boundary vertices. The following is a direct corollary of \Cref{lm:pattern-bound-undir}.

\begin{corollary}\label{cor:pattern-count}$|P_R| =  \tilde{O}(r^{3(h-1)/2})$ for every $R\in \mathcal{R}$.
\end{corollary}
\begin{proof}
	By \Cref{lm:pattern-bound-undir}, the number of patterns is $O((|\partial R|\cdot |V(R)|)^{h-1}) = \tilde{O}(r^{3(h-1)/2})$.
\end{proof}

Next, we bound the running time of (Step 1). 

\begin{lemma}\label{lm:diam-Step1} Given $D(B,V)$, we can implement (Step 1) in $\tilde{O}(n^2/\sqrt{r})$ time. 
\end{lemma}
\begin{proof}
	First, by \Cref{lm:r-division}, $\mathcal{R}$ can be constructed in time $O(n^{1+\eps}\sqrt{r})$ for any fixed constant $\eps > 0$. As $h\geq 4$,  $n^2/\sqrt{r} = n^{2-1/(3h-1)} = \Omega(n^{1.5})$ while $n^{1+\eps}\sqrt{r} = n^{1+\eps + 1/(3h-1)} = n^{1.2 + \eps}$. Thus, by choosing $\eps = 0.1$, we have  $n^{1+\eps}\sqrt{r} = O(n^2/\sqrt{r})$. That is, $\mathcal{R}$ can be constructed in $\tilde{O}(n^2/\sqrt{r})$ time. 
	
	Next we compute $P_R$, which is initialized to be $\emptyset$. Then for each $u \in V$, we look up the distance from all vertices of $\partial R$ to $u$ in $D(B,V)$. Then we compute the pattern $\patt_u$ from $u$ to $R$, in $O(|\partial R|)$ time. We then add $\patt_u$ to $P_R$ if $\patt_u$ is currently not in $P_R$; this check can be done in $O(|\partial R|)$ time using a trie data structure, say. The total running time to compute $P_R$ is $O(n|\partial R|)$. Thus, the total running time of this step is $O(n \sum_{R\in \mathcal{R}}|\partial R|) = \tilde{O}(n^2/\sqrt{r})$ by \Cref{lm:r-division}.   
\end{proof}

\begin{lemma}\label{lm:diam-Step2}  Given $D(B,V)$ and $\{P_R\}_{R\in \mathcal{R}}$, we can implement (Step 2) in $\tilde{O}(n r^{(3h-2)/2})$ time.
\end{lemma}
\begin{proof}
	For each pattern $\patt\in P_R$, we can compute the distance $d(\patt,v)$ for each $v\in V(R)$ in time $O(|\partial R|) = \tilde{O}(\sqrt{r})$. Thus, finding the furthest vertex from $\patt$ takes $\tilde{O}(\sqrt{r}|V(R)|)$ time. By \Cref{cor:pattern-count}, the running time to compute the table $T^{(2)}_R$ is $\tilde{O}(r^{1/2}|V(R)| r^{3(h-1)/2}) = \tilde{O}(r^{(3h-2)/2})|V(R)|$. Thus, the total running time of (Step 2) is $\tilde{O}(r^{(3h-2)/2})\sum_{R\in \mathcal{R}}|V(R)|) = \tilde{O}(nr^{(3h-2)/2})$, as claimed.
\end{proof}

\begin{lemma}\label{lm:diam-Step3}  (Step 3) can be implemented in $\tilde{O}(n^2/\sqrt{r})$ time given the information computed in (Step 1) and (Step 2).
\end{lemma}
\begin{proof} First we bound the running time to compute $\ecc(u)$ for a given vertex $u\in V$.  For the cluster $R\in \mathcal{R}$ such that $u\in R$, computing $\Delta(u,R)$ takes $O(|V(R)|) = O(r)$ time. If $u\not\in R$, we can look up (the ID of) the pattern $\patt_u$ in $T^{(1)}_R$ in $O(1)$ time. Given $\patt_u$, we can lookup $d(\patt_u,v)$ in $O(1)$ time from $T^{(2)}_R$ constructed in (Step 2).  Furthermore, $d_G(u,s_0)$ can be found  directly from $D(B,V)$ in $O(1)$ time. Thus, the running time to compute $\Delta(u,R)$ is $O(1)$. We conclude that the total running time to compute $\Delta(u,R)$ for all $R$ is $O(r + |\mathcal{R}|)$. By \Cref{lm:r-division}, 
	\begin{equation*}
		O(r + |\mathcal{R}|) = \tilde{O}(r + n/\sqrt{r}) = \tilde{O}(n/\sqrt{r})
	\end{equation*}
	as $r = n^{2/(3h-1)} \leq n^{0.2}$ with $h\geq 4$.  This means that the total running time to compute all the eccentricities is $\tilde{O}(n^2/\sqrt{r})$.
\end{proof}

By \Crefrange{lm:diam-Step1}{lm:diam-Step3}, the total running time to compute all the eccentricities (and hence the diameter) of $G$ is:

\begin{equation}
	\tilde{O}(\frac{n^2}{\sqrt{r}} + n r^{(3h-2)/2}) = \tilde{O}(n^{2-\frac{1}{3h-1}})
\end{equation}
when $r = n^{2/(3h-1)}$.

\subsubsection{Wiener Index}\label{subsec:Wiener}

We show how to compute Wiener index in truly subquadratic time as described in \Cref{cor:dimater-unweighted}. For any two set of vertices $X,Y\subseteq V$, let $W(X,Y)  = \sum_{x\in X}\sum_{y\in Y}d_G(x,y)$. The Wiener index of $G$ is $\frac{1}{2}W(V,V)$, and thus our goal is to compute $W(V,V)$. Let $\mathcal{R}$ be an $r$-division of $G$ computed by \Cref{lm:r-division} for  $r = n^{2/(3h-1)}$.  Let $R^{\circ} = V(R)\setminus \partial R$. Recall that in \Cref{subsec:diameter-unweighted}, we define  $B = \cup_{R\in \mathcal{R}}\partial R$.  Observe that:
\begin{equation}\label{eq:W-VV}
	\begin{split}
		W(V,V) &= W(B,V) + \sum_{R\in \mathcal{R}} W(R^{\circ}, V)
		\\&=  W(B,V) + \sum_{R\in \mathcal{R}}W(R^{\circ}, V(R)) + \sum_{R\in \mathcal{R}}W(R^{\circ}, V\setminus V(R))
	\end{split}
\end{equation}

First, we focus on computing $ W(B,V)$ and $\sum_{R\in \mathcal{R}}W(R^{\circ}, V(R))$.

\begin{lemma}\label{lm:WVB} $W(B,V)$ can be computed in time $\tilde{O}(n^2/\sqrt{r})$ and $\sum_{R\in \mathcal{R}}W(R^{\circ}, V(R))$ can be computed in time $O(nr)$.
\end{lemma}
\begin{proof}
	By \Cref{obs:Tcal-Time}, all the distances from vertices in $B$ to vertices in $V$ (the set $D(B,V)$) can be computed in time $\tilde{O}(n^2/\sqrt{r})$, which also is the running time to compute  $W(B,V)$.
	
	By \Cref{lm:r-division}, each cluster has a size at most $r$. Furthermore, computing the distance from a vertex in $R$ to all other vertices $R$ can be done in $O(r)$ time using BFS. Thus, the running time to compute $\sum_{R\in \mathcal{R}}W(R^{\circ}, R)$ is $\sum_{R\in \mathcal{R}} O(r |R^{\circ}|)  = O(nr)$.
\end{proof}

Next, we bound the running time to compute $\sum_{R\in \mathcal{R}}W(R^{\circ}, V\setminus V(R))$.

\begin{lemma}\label{lm:WRV} $\sum_{R\in \mathcal{R}}W(R^{\circ}, V\setminus V(R))$ can be computed in time $\tilde{O}(n r^{(3h-2)/2} + n^2/\sqrt{r})$,
\end{lemma}
\begin{proof} First we compute the set of patterns $\{P_R\}_{R\in \mathcal{R}}$ of all clusters in $\mathcal{R}$ in time  $\tilde{O}(n^2/\sqrt{r})$ by \Cref{lm:diam-Step1}. Next, we observe that:
	\begin{equation}\label{eq:WRV}
		\begin{split}
			\sum_{R\in \mathcal{R}}W(R^{\circ}, V\setminus V(R)) = 	\sum_{R\in \mathcal{R}} \sum_{u\in V\setminus V(R)} W(u,R^{\circ})
		\end{split}
	\end{equation}
	Furthermore, by \Cref{lm:dist-via-pattern},
	\begin{equation*}
		W(u,R^{\circ}) = \sum_{v\in R^{\circ}} d_G(u,v)  = |R^{\circ}|d_G(u,s_0) + \sum_{v\in R^{\circ}} d(\patt_u,v)
	\end{equation*}
	where $s_0$ is the first vertex in the boundary sequence $\sigma_R$ of cluster $R$. The distance $d_G(u,s_0)$ is already computed, i.e, $d_G(u,s_0)\in D(B,V)$. 
	
	In \Cref{lm:diam-Step2}, we find the furthest vertex from each pattern $\patt \in P_R$ by iterating over all vertices of $R$ in total time  $\tilde{O}(\sqrt{r}|V(R)|)$ time. Thus, we can compute the sum $\sum_{v\in R^{\circ}} d(\patt,v)$ in time $\tilde{O}(\sqrt{r}|V(R)|)$, and running time for to compute all the sums of \emph{all patterns} in $P_R$ is $\tilde{O}(\sqrt{r}|V(R)|r^{3(h-1)/2}) = \tilde{O}(r^{(3h-2)/2})|V(R)|$. We can think of this as preprocessing time for computing $	W(u,R^{\circ}) $. Over all clusters in $\mathcal{R}$, the total preprocessing time is:
	\begin{equation}\label{eq:processing-time}
		\sum_{R\in \mathcal{R}}\tilde{O}(r^{(3h-2)/2})|V(R)| = \tilde{O}(nr^{(3h-2)/2})
	\end{equation}
	which is the first term in the running time. 
	
	Once the sums of distances for all patterns in $P_R$ are given, we can store them in a table keyed by the ID of the patterns, and then we can look up $\sum_{v\in R^{\circ}} d(\patt_u,v)$ in $O(1)$ time. As a result, we can compute  $	W(u,R^{\circ})$ in $O(1)$ time, and hence by \Cref{eq:WRV}, $W(R^{\circ}, V\setminus V(R))$ can be computed in time:
	\begin{equation*}
		\sum_{R\in \mathcal{R}} O(n) = O(n |\mathcal{R}|) = \tilde{O}(n^2/\sqrt{r})
	\end{equation*}
	by \Cref{lm:r-division}, which is the second term in the running time.
\end{proof}

By \Cref{lm:WVB} and \Cref{lm:WRV}, the total running time to compute $W(V,V)$ is $\tilde{O}(\frac{n^2}{\sqrt{r}} + n r^{(3h-2)/2}) = \tilde{O}(n^{2-\frac{1}{3h-1}})$ when  $r = n^{2/(3h-1)}$ as claimed in \Cref{cor:dimater-unweighted}.

\subsubsection{Exact Distance Oracle}\label{subsec:oracle-undirected}

We construct the first exact distance oracle for unweighted minor-free graphs with subquadratic space-query time trade-off and subquadratic preprocessing time as described in \Cref{cor1:oracle-unweighted}.

\paragraph{Construction.} The construction has two steps:

\begin{itemize}
	\item \textbf{(Step 1).~} Construct an $r$-division $\mathcal{R}$ of $G$ with $r = n^{2/(3h-1)}$, and for each cluster $R\in \mathcal{R}$, store a set of patterns $P_R$ w.r.t an (arbitrary) sequence of boundary vertices $\sigma_R$ in a table $T_R$. We also store the exact distances of all pairs of vertices in $R$.
	\item \textbf{(Step 2).~} For each vertex $u$ and a region $R\in \mathcal{R}$: (2a) if $u\in R$, we store $d(\patt,u)$ for every pattern $\patt \in P_R$; (2b) if $u\not\in R$, we store a pointer from $u$ to  its pattern $\patt_u$ in table $T_R$ and the distance from $u$ to the first vertex in the sequence of boundary vertices $\sigma_R$. 
\end{itemize}

\paragraph{Querying distances.~} Given two vertices $u$ and $v$, if there is a region $R$ containing both $u$ and $v$, we can simply look up their distance stored at $R$ in $O(1)$ time. Otherwise, let $R_v$ be the region containing $v$.  First, we look up the distance from $u$ to the first vertex in the boundary sequence $\sigma_{R_v}$, say $s_0$, in $O(1)$ time. Then, we look up the pattern $\patt_u \in P_{R_{v}}$ of $u$ in $R_v$, and the distance $d(v, \patt_u)$ in total $O(1)$ time due to the construction in (Step 2). Finally, we return:
\begin{equation}\label{eq:exact-dist-uw}
	d_G(u,s_0) + d(\patt_u,v)~.
\end{equation}
\Cref{lm:dist-via-pattern} implies that the returned distance is $d_G(u,v)$. The total query time is $O(1)$. 

\paragraph{Space analysis.} By \Cref{cor:pattern-count}, the number of patterns is $\tilde{O}(r^{3(h-1)/2})$ and by \Cref{lm:r-division}, $|\mathcal{R}| =  \tilde{O}(n/\sqrt{r})$. The total space of (Step 1) and (Step 2(a)) is:
\begin{equation}
	\tilde{O}(\sum_{R\in \mathcal{R}}(r^{3(h-1)/2}|V(R)| + |V(R)|^2)) = \tilde{O}(\sum_{R\in \mathcal{R}}(r^{3(h-1)/2}|V(R)| + r |V(R)|)) = \tilde{O}(nr^{3h/2 - 2})
\end{equation}
as $h\geq 4$. The total space of Step 2(b) is $O(n|\mathcal{R}|) = \tilde{O}(n^2/\sqrt{r})$ by \Cref{lm:r-division}. Thus, the total space of the oracle is:

\begin{equation}
	\tilde{O}(nr^{3h/2 - 2} + n^2/\sqrt{r}) = \tilde{O}(n^{2 - \frac{1}{3h-1}})
\end{equation}
with $r = n^{2/(3h-1)}$.

\paragraph{Construction Time.} We observe that the amount of information we need to construct the distance oracle is exactly the amount of information we need to compute the diameter and the Wiener index. Thus, the running time to compute all the information is  $\tilde{O}(n^{2-\frac{1}{3h-1}})$. 

\begin{remark}\label{rm:space} We can further reduce the space of the oracle by increasing the construction time by choosing $r$ differently, or by increasing the query time using the nested $r$-division following the line of reasoning in~\cite{FMW20}.
\end{remark}

\subsubsection{Other Applications}\label{subsec:other-undirected}

Here we discuss other algorithmic applications of our \Cref{thm:LW-undirected}. 

\paragraph{Metric compression.~} Li and Parter~\cite{LP19} showed that for any two sets of vertices $S,T$ in an unweighted \emph{planar} graph of diameter $D$ such that $S$ is on the boundary of the outer face of the graph, then one can compress all the distances from $T$ to $S$ using only $\tilde{O}(|S|^3 D  + |T| |\log(|S| D)|)$ bits. Here it is instructive to think of a canonical regime where $D$ is a constant and $1\ll \poly(|S|) \ll |T|$. In this canonical regime, the compression scheme has space $\tilde{O}(\poly(|S|)  + |T|)$ bits instead of $\tilde{O}(|S|\cdot |T|)$ bits by simply storing all the distances from $T$ to $S$. This compression scheme has an application in computing the diameter of planar graphs in the distributed CONGEST model. 

Using \Cref{thm:LW-undirected}, we improve the compression scheme by Li and Parter in two aspects: $S$ is no longer restricted, and $G$ could be any minor-free graphs. In the canonical regime, the space of our compression scheme is $\tilde{O}(\poly(|S|) +  |T|)$ which is the same as Li-Parter space bound up to a  factor of $\poly(|S|)$ in the additive term.

Now we give a more formal description of our result. Let $S = \{s_0,s_1,\ldots, s_{k-1}\}$ and $T= \{t_1,t_2,\ldots, t_{\ell}\}$ where $k = |S|$ and $\ell = |T|$. For each vertex $v\in V$, we define $$\tuple(v) = \langle d_G(v,s_0), \ldots, d_G(v,s_{k-1})\rangle~,$$ which is called a \emph{distance tuple} of $v$ w.r.t $S$. Li and Pater showed in their Theorem 2.2~\cite{LP19} that the set $\tuple(V) = \{\tuple(v): v\in V\}$ has size $O(|S^3|D)$ when $S$ is on the outer face of a planar graph of diameter $D$. Hence, to compress the distances from $T$ to $S$, one only needs to store $\tuple(V)$ using $\tilde{O}(|S|^3D)$ bits and then for each $t\in T$, one stores a pointer from $t$ to its corresponding distance tuple $\tuple(t)\in \tuple(V)$. Here, we show that in our more general setting where $S$ has no restriction and $G$ is $K_h$-minor-free, the number of tuples is bounded by $O((|S|\cdot D)^{O(h)}) = O(\poly(|S|))$ for fixed $h$ and $D$. This implies our result on the metric compression.

\begin{lemma}\label{lm:distance-tuple} $|\tuple(V)|  = O(|S|^{h-1}\cdot D^h)$ when $G$ is a $K_h$-minor-free  graph and has diameter at most $D$.
\end{lemma}
\begin{proof}
	The proof is the same as the proof of \Cref{lm:pattern-bound-undir}. The only difference is that now  $ -D \leq d_G(v,s_i)- d_G(v,s_0) \leq D$ by the triangle inequality. Let's fix the distance from $d_G(v,s_0)$, and $M = \{-D,\ldots, -1,0,1,\ldots D\}$. Let $\bar{p}_v$ be a set obtain by adding pairs  $(i,\Delta)$ for $\Delta \in M$ such that $d_G(v,s_i)-d_G(v,s_{i-1})\leq \Delta$  to the set $\bar{p}_v$ for each $i\in [1,k-1]$. Then  there is a bijection between the set $\{\bar{p}_v\}_{v\in V}$ and  $\wlp_{G,M}$.  By the Sauer–Shelah Lemma (\Cref{lm:SS}), we have  $|\wlp_{G,M}| = O((|S||M|)^{h-1}) = O((|S|\cdot D)^{h-1})$. As we have $D$ choices for $d_G(v,s_0)$, the number of different distance tuples is at most  $O((|S|\cdot D)^{h-1} \cdot D) = O(|S|^{h-1}D^h)$.  
\end{proof}

\paragraph{Computing diameter and all-vertices eccentricities in low-treewidth minor-free graphs.~} Abboud, Williams, and Wang~\cite{AWW16} studied the problem of computing diameter in unweighted graphs of treewidth $k$. They showed surprisingly that, there exists a constant $c > 0$ such that for any $k \leq c\cdot\log n$, under the Strong Exponential Time Hypothesis (SETH), there is no algorithm with running time $n^{2-\eps}2^{\Omega(k)}$ to compute the diameter for any fixed $\eps > 0$. That is, if one insists on having an algorithm with truly subquadratic time, one has to pay an \emph{exponential dependency} on the treewidth. They also presented an algorithm for distinguishing diameter 2 vs diameter 3 graphs with running time $n^{1+o(1)}2^{O(k\log k)}$. Husfeldt~\cite{Husfeldt17} designed an improved algorithm with running time $O(d^{O(k)}n)$ where $d$ is the diameter using dynamic programming. An open question is to design in algorithm with running time $O(d^{O(1)}2^{O(k)}n)$. 

We show that if the input graph $G$ has treewidth $k$, and in addition, is $K_h$-minor-free, for a fixed constant $h$, then one can find the diameter of $G$ in time $O((dk)^{O(1)} n)$. Notably, the dependency on the treewidth $k$ is polynomial instead of exponential. We note that the class of $K_h$-minor-free graphs of treewidth $k$ includes well-studied classes of graphs, such as  $k$-outerplanar graphs,  Halin graphs, and series-parallel graphs.

Here, we sketch our argument. The basic idea is to use \Cref{thm:LW-undirected} to optimize the running time of the dynamic programming algorithm by Husfeldt~\cite{Husfeldt17} (for computing all-vertices eccentricities and hence diameter). For each bag $B = \{s_0,s_1,\ldots, s_{k-1}\}$ of size $k$ in the tree decomposition, the dynamic program keeps track of all the distance tuples of vertices in the graph induced by vertices in descendant bags of $B$ (and including $B$). The maximum number of distance tuples is $d^{O(k)}$, which results in running time  $d^{O(k)}n$. When $G$ is $K_h$-minor-free, then by \Cref{lm:distance-tuple}, the number of distance tuples is $\poly(k\cdot d)$, and hence the running time of the dynamic program becomes  $\poly(k\cdot d)n$.

\paragraph{Approximate distance oracles in planar graphs.~} In \cite{Le22}, Le constructed a $(1+\eps)$-approximate distance oracle for planar graphs with $\tilde{O}(n/\eps^{o(1)})$ space and $\tilde{O}(1)$ query time. That is, the space-query product trade-off depends sublinearly on $1/\eps$. A key ingredient of the construction is a polynomial bound on the number of (approximate) distance tuples by Li and Parter~\cite{LP19}. Our set system $\wlp_{G,M}$ also gives a polynomial bound on the number of such distance tuples and hence could be used in the same way to derive the result in  \cite{Le22}.

%% file: LPGen.tex
\subsection{VC dimension of $\wlp_{G,M}$}\label{subsec:Vc-dim-LPhat}

In this section, we fix $G = (V,E)$ to be an undirected $K_h$-minor-free graph. We first prove \Cref{thm:LW-undirected}, which we restate below.

\LWUndir*

Our proof is by contradiction.  Suppose that there is a set $Y$ of size $h$ that is shattered by  $\wlp_{G,M}$. W.l.o.g., we assume that $Y = \{(1,\Delta_{1}), \ldots, (h,\Delta_{h})\}$. We first observe that:  

\begin{observation}\label{obs:s-diff} $s_{i} \not= s_{j}$ for any $1\leq i\not=j \leq h$.
\end{observation}
\begin{proof} Suppose otherwise, that $s_{i} = s_{j}$. W.l.o.g, we assume that $\Delta_{i} < \Delta_{j}$. This means if $(s_{i},\Delta_{i}) \in \widehat{X}_v$ for some vertex $v\in V$, then $(s_{i},\Delta_{j}) \in \widehat{X}_v$, since $d_G(v,s_{i})\leq d_G(v,s_{0}) + \Delta_{i}$ implies that $d_G(v,s_{i})\leq d_G(v,s_{0}) + \Delta_{j}$. However, since $\wlp_{G,M}$ shatters $Y$, by  definition of shattering, there exists a set $\widehat{X}_v$ containing $(s_{i},\Delta_{i})$ but not $(s_{i},\Delta_{j})$, a contradiction.
\end{proof}

For every two elements $(s_{i},\Delta_{i})$ and $(s_{j},\Delta_{j})$ with $i\not=j$ in $Y$, let $v_{ij}$ be a vertex such that $\{(s_{i},\Delta_{i}), (s_{j},\Delta_{j})\} = \widehat{X}_{v_{ij}}\cap Y$.

Let $\widehat{G}$ be a graph obtained from $G$ by adding (tiny) perturbed weights to edges of $G$ in such a way that (i) shortest paths in $\widehat{G}$ between vertices are unique and (ii) every shortest path in $\widehat{G}$ is also a shortest path in $G$. (Some shortest path in $G$ may no longer be a shortest path in  $\widehat{G}$.) We can think of $\widehat{G}$ as providing a tie-breaking scheme for shortest paths in $G$. The perturbation exists by the Isolation Lemma~\cite{VV86}.

\begin{definition}\label{def:tij} We define vertex $t_{ij}$ to be the vertex in $\pi(v_{ij},s_{i},\widehat{G})\cup \pi(v_{ij}, s_{j},\widehat{G})$ such that:
	\begin{itemize}
		\item[(a)] $d_G(v_{ij}, t_{ij}) + d_G(t_{ij},s_{i}) \leq \Delta_{i} + d_G(v_{ij}, s_0)$ and $d_G(v_{ij}, t_{ij}) + d_G(t_{ij},s_{j}) \leq \Delta_{j} + d_G(v_{ij}, s_0)$.
		\item[(b)]  the sum of distance $d_{\widehat{G}}(t_{ij},s_{i})  + d_{\widehat{G}}(t_{ij},s_{j}) $ is minimum.
	\end{itemize}
\end{definition}

Note that the distances in Item (a) of \Cref{def:tij} are w.r.t. graph $G$ while the distances in Item (b) are w.r.t. $\widehat{G}$. By the definition of $\widehat{G}$,  $\pi(s_{i},v_{ij},\widehat{G})$ is also a shortest path in $G$; sometimes we abuse notation by using $\pi(s_{i},v_{ij},\widehat{G})$ to refer to its corresponding shortest path in $G$.  We remark that $t_{ij}$ exists since $v_{ij}$ is a possible choice for $t_{ij}$ satisfying (a). A good, but not accurate,  interpretation of $t_{ij}$ to keep in mind is that when the two shortest paths $\pi(s_{i},v_{ij},\widehat{G})$ and $\pi(s_{i},v_{ij},\widehat{G})$ shares the same vertex other than $v_{ij}$, then $t_{ij}$ is the common vertex furthest from $v_{ij}$; this would be the case if we restrict $t_{ij}$ to be in $\pi(v_{ij},s_{i},\widehat{G})\cap \pi(v_{ij}, s_{j},\widehat{G})$ instead of $\pi(v_{ij},s_{i},\widehat{G})\cup \pi(v_{ij}, s_{j},\widehat{G})$  as in \Cref{def:tij}. Indeed, the role of $t_{ij}$ in the proof is subtler than just being the furthest common vertex.

\begin{figure}[!htb]
	\centering{\includegraphics[width=.8\textwidth]{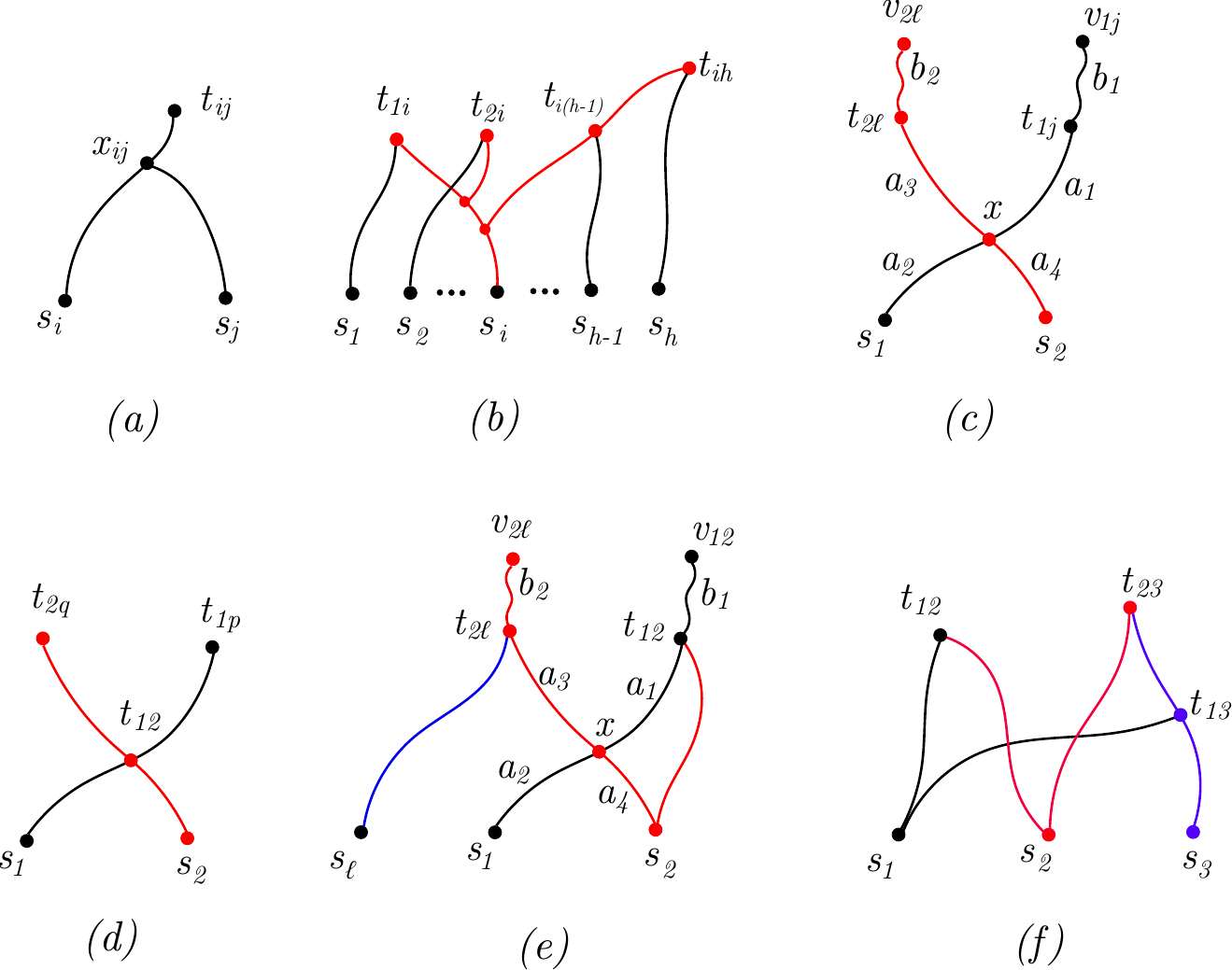}}
	\caption{(a)  $x_{ij} \in (\pi(s_{i},t_{ij},G)\cap \pi(s_{j},t_{ij},G))\setminus \{t_{ij}\}$; (b) the bunch of $s_{i}$ contains all red paths, $t_{i(h-1)}$ is an internal vertex of $\bunch(s_i)$ but an endpoint of $\bunch(s_{h-1})$; (c) $x\in \pi(s_{1},t_{1j},G)\cap  \pi(s_{2},t_{2k},G)$; (d) Illustration for the proof that $t_{12}$  must be an endpoint; (e) $\{x\}\in  \pi(s_{1},t_{12},G) \cap \pi(s_{2},t_{2\ell},G)$; (f) A $K_3$ minor constructed from three bunches  $\bunch(s_1)$,  $\bunch(s_2)$, and  $\bunch(s_3)$; $t_{12}$ will be contracted to $s_{1}$, $t_{23}$ is contracted to $s_{2}$, and $t_{13}$ is contracted to $s_{3}$.  
	}\label{fig:VCGen}
\end{figure}

\begin{claim}\label{obs:disjoint}
	$\pi(s_{i},t_{ij},G)$ and $\pi(s_{j},t_{ij},G)$ are internally disjoint.
\end{claim}
\begin{proof}
	Suppose otherwise; there would be a vertex $x_{ij} \in (\pi(v_i,t_{ij},G)\cap \pi(v_j,t_{ij},G))\setminus \{t_{ij}\}$; see \Cref{fig:VCGen}(a). Then, we have:
	\begin{equation*}
		\begin{split}
			d_G(v_{ij}, x_{ij}) + d_G(x_{ij},s_{i}) &\leq d_G(v_{ij}, t_{ij}) + d_G(t_{ij},x_{ij}) + d_G(x_{ij},s_{i}) \\
			&= d_G(v_{ij}, t_{ij}) + d_G(t_{ij},s_{i}) \\
			&\leq  \Delta_{i} + d_G(v_{ij}, s_0) \qquad \mbox{(by Item (a) in \Cref{def:tij})}
		\end{split}
	\end{equation*}
	By the same argument, we have:
	\begin{equation*}
		d_G(v_{ij}, x_{ij}) + d_G(x_{ij},s_{j})\leq \Delta_{j} + d_G(v_{ij}, s_0).
	\end{equation*}
	which means $x_{ij}$ satisfies Item (a) in \Cref{def:tij}. Furthermore, as $x_{ij} \in (\pi(v_i,t_{ij},G)\cap \pi(v_j,t_{ij},G))\setminus \{t_{ij}\}$, we have:
	\begin{equation*}
		\begin{split}
			d_{\widehat{G}}(x_{ij},s_{i}) +  d_{\widehat{G}}(x_{ij},s_{j}) &=  d_{\widehat{G}}(t_{ij},s_{i}) - d_{\widehat{G}}(t_{ij},x_{ij}) +  d_{\widehat{G}}(t_{ij},s_{j})- d_{\widehat{G}}(t_{ij},x_{ij}) \\
			&< d_{\widehat{G}}(t_{ij},s_{i}) +  d_{\widehat{G}}(t_{ij},s_{j})
		\end{split}
	\end{equation*}
	contradicting the minimality of $t_{ij}$ by Item (b) in \Cref{def:tij}.
\end{proof}

We define the \emph{bunch of each vertex} $s_{i}$ (see \Cref{fig:VCGen}(b)):

\begin{equation}\label{eq:bunch2}
	\bunch(s_{i}) = \cup_{j\not=i} \pi(s_{i},t_{ij},\widehat{G})
\end{equation}

We say that $t_{ij}$ is an \emph{endpoint} of $\bunch(s_{i})$ if it has degree 1 in the subgraph $\bunch(s_{i})$. Otherwise, we say that $t_{ij}$ is an internal vertex of  $\bunch(s_{i})$. (One case where  $t_{ij}$ is not an endpoint is when the path $\pi(s_{i},t_{ij},\widehat{G})$ is a subpath of $\pi(s_{i},t_{ij'},\widehat{G})$ of another vertex $t_{ij'}$.)

\begin{lemma}\label{lm:bunch-disjiont2} For every $a\not=b$, $\bunch(s_{a})\cap \bunch(s_{b}) = \{t_{ab}\}$. Furthermore, $t_{ab}$ is either an endpoint of $\bunch(s_{a})$, or an endpoint of $\bunch(s_{b})$, or both.
\end{lemma}
\begin{proof} By the symmetry of $s_{i}$, we prove the lemma for $a = 1,b = 2$. Let $\pi(s_{1},t_{1j},\widehat{G})$ and $\pi(s_{2},t_{2\ell},\widehat{G})$ be paths in  $\bunch(s_{1})$ and $\bunch(s_{2})$, respectively.  
	
	\begin{claim}\label{clm:jk-12-dist} If $\{j,\ell\}\cap\{1,2\} = \emptyset$, then $\pi(s_{1},t_{1j},\widehat{G})\cap  \pi(s_{2},t_{2\ell},\widehat{G}) = \emptyset$.
	\end{claim}
	\begin{proof}
		Suppose otherwise, there exists $x\in \pi(s_{1},t_{1j},\widehat{G})\cap  \pi(s_{2},t_{2k},\widehat{G})$; see \Cref{fig:VCGen}(c). Let:
		\begin{equation*}
			\begin{split}
				a_1 = d_G(t_{1j},x)&\qquad a_2 = d_G(x,s_{1})\\
				a_3 = d_G(t_{2\ell},x)&\qquad a_4 = d_G(x,s_{2})\\
				b_1 = d_G(v_{1j}, t_{1j}) &\qquad b_2 = d_G(v_{2\ell}, t_{2\ell}) 
			\end{split}
		\end{equation*}
		
		By definition of $t_{1j}$ (Item (a) in \Cref{def:tij}), it holds that $b_1 + a_1+a_2\leq \Delta_{1} + d_G(v_{1j}, s_{0})$. For the same reason, we have $b_2 + a_3+a_4\leq \Delta_{2} + d_G(v_{2k}, s_0)$.  It follows that:
		\begin{equation}\label{eq:a1234-sum2}
			a_1 + a_2 + a_3+ a_4 + b_1 + b_2\leq   \Delta_{1} + \Delta_{2} +  d_G(v_{1j}, s_0) + d_G(v_{2\ell}, s_0)
		\end{equation}
		On the other hand, $(s_{2}, \Delta_{2})\not\in  \widehat{X}_{v_{1j}}$. Thus, $d_{G}(v_{1j},s_{2}) > \Delta_{2} + d_{G}(v_{1j},s_0) $. By the triangle inequality, $b_1 + a_1+a_4\geq d_{G}(v_{1j},s_{2}) $, which implies that:
		\begin{equation}\label{eq:subeq-114}
			b_1 + a_1+a_4  > \Delta_{2} + d_{G}(v_{1j},s_0)
		\end{equation}
		By the same argument, we have that $b_2 + a_2+a_3  > \Delta_{1} + d_{G}(v_{2\ell},s_0) $.  Combining with \Cref{eq:subeq-114}, we get:
		\begin{equation}\label{eq:a1234-sum-final}
			a_1 + a_2 + a_3+ a_4 + b_1 + b_2 >  \Delta_{1} + \Delta_{2} +  d_G(v_{1j}, s_0) + d_G(v_{2\ell}, s_0),
		\end{equation}
		which contradicts \Cref{eq:a1234-sum2}. Thus, $x$ does not exist.
	\end{proof}

	\begin{claim}\label{clm:jk-12-b2} If $\{j,\ell\}\cap\{1,2\}  \not=\emptyset$, then $\pi(s_{1},t_{1j},\widehat{G})\cap  \pi(s_{2},t_{2\ell},\widehat{G})\subseteq \{t_{12}\}$, and that $t_{12}$ is either an endpoint of $\bunch(s_{1})$ or $\bunch(s_{2})$ or both.
	\end{claim}
	\begin{proof} If $t_{12}\in \pi(s_{1},t_{1j},\widehat{G})\cap  \pi(s_{2},t_{2\ell},\widehat{G})$, then we claim that $t_{12}$ must be an endpoint of $\bunch(s_{1})$ or $\bunch(s_{2})$ or both. Suppose otherwise; that is $t_{12}$ is not an endpoint of either $\bunch(s_{1})$ or $\bunch(s_{2})$. It means that there are two paths $\pi(s_{1},t_{1p},\widehat{G})$ and $\pi(s_{2},t_{2q},\widehat{G})$ such that $\{t_{12}\} \subseteq \pi(s_{1},t_{1p},\widehat{G})\cap  \pi(s_{2},t_{2q},\widehat{G})$ and that $\{p,q\}\cap\{1,2\} = \emptyset$; see \Cref{fig:VCGen} (d). The existence of such two paths contradicts \Cref{clm:jk-12-dist}.
		
	We are now proving  that $\pi(s_{1},t_{1j},\widehat{G})\cap  \pi(s_{2},t_{2\ell},\widehat{G})\subseteq \{t_{12}\}$. Observe that if $j=2$ and $\ell=1$, then $t_{12}= \pi(s_{1},t_{1j},\widehat{G}) \cap \pi(s_{2},t_{2\ell},\widehat{G})$ by \Cref{obs:disjoint}. Thus, \Cref{clm:jk-12-b2} follows. It remains to consider two other cases: (i) $j=2, \ell\not=1$ or (ii) $j\not=2, \ell=1$. Both cases are symmetric, and hence w.l.o.g, we only consider case (i). 
		
		Suppose that there exists $x\not= t_{12}$ such that  $\{x\}\in  \pi(s_{1},t_{12},\widehat{G}) \cap \pi(s_{2},t_{2\ell},\widehat{G})$ ($j=2$ now); see \Cref{fig:VCGen}(e). We define $a_1,a_2,a_3,a_4,b_1,b_2$ as in \Cref{clm:jk-12-dist}, specifically:
		\begin{equation*}
			\begin{split}
				a_1 = d_G(t_{12},x)&\qquad a_2 = d_G(x,s_{1})\\
				a_3 = d_G(t_{2\ell},x)&\qquad a_4 =  d_G(x,s_{2})\\
				b_1 = d_G(v_{12},x) &\qquad b_2 = d_G(v_{2\ell},x)
			\end{split}
		\end{equation*}
		By definition of $t_{12}$ (Item (a) in \Cref{def:tij}), $b_1 + a_1 + a_2 \leq \Delta_{1} + d_G(v_{12},s_0)$. By the same argument,  $b_2 + a_3 + a_4 \leq \Delta_{2} + d_G(v_{2\ell},s_0)$. Thus, 
		\begin{equation}\label{eq:a1234-sum3}
			b_1 + b_2 + a_1 + a_2  + a_3+a_4 \leq \Delta_{1} + \Delta_{2} + d_G(v_{12},s_0) + d_G(v_{2\ell},s_0).
		\end{equation}
		
		Since  $(s_{1},\Delta_{1})\not\in \widehat{X}_{v_{2\ell}}$,  $d_G(v_{2\ell}, s_{1}) > \Delta_{1} + d_{G}(v_{2\ell},s_0)$. By the triangle inequality, we have that $b_2 + a_3 + a_2 > \Delta_{2} + d_G(v_{2k},s_0)$. Thus, by \Cref{eq:a1234-sum3}, $b_1 +  a_1 + a_4  \leq \Delta_{2} +  d_G(v_{12},s_0)$.  In summary, we have:
		
		\begin{equation*}
			\begin{split}
				b_1 + a_1 + a_2 &\leq \Delta_{1} + d_G(v_{12},s_0)\\
				b_1 +  a_1 + a_4  &\leq \Delta_{2} +  d_G(v_{12},s_0)\\
			\end{split}
		\end{equation*}
		By the triangle inequality, we have that $d_G(v_{12},x)\leq b_1 + a_1$. Recall that $a_2 = d_G(x,s_{1})$ and $a_4 = d_G(x,s_{2})$.  It follows that:
		\begin{equation}\label{key}
			\begin{split}
				d_G(v_{12},x) + d_G(x,s_{1}) &\leq \Delta_{1} + d_G(v_{12},s_0)\\
				d_G(v_{12},x) + d_G(x,s_{2})  &\leq \Delta_{2} +  d_G(v_{12},s_0)\\
			\end{split}
		\end{equation}
		
		Thus, $x$ satisfies Item (a) of \Cref{def:tij}. We now show that $d_{\widehat{G}}(x, s_{1}) + d_{\widehat{G}}(x, s_{2}) < d_{\widehat{G}}(t_{12}, s_{1}) + d_{\widehat{G}}(t_{12}, s_{2})$, which will give a contradiction by the choice of $t_{12}$ in Item (b) of \Cref{def:tij}. 
		
		Observe that by a triangle inequality, $d_G(x,s_{2}) \leq d_G(x, t_{12}) + d_G(t_{12}, s_{2})$. Since shortest paths are unique in $\widehat{G}$, $d_{\widehat{G}}(x,s_{2}) < d_{\widehat{G}}(x, t_{12}) + d_{\widehat{G}}(t_{12}, s_{2})$. (See \Cref{fig:VCGen}(e).) Since $\pi(t_{12}, s_{1},\widehat{G})[x,s_{1}]$ is a shortest path in $\widehat{G}$ , $d_{\widehat{G}}(x,s_{1}) = d_{\widehat{G}}(t_{12}, s_{1}) - d_{\widehat{G}}(x, t_{12})$. It follows that:
		
		\begin{equation*}
			\begin{split}
				d_{\widehat{G}}(x, s_{1}) + d_{\widehat{G}}(x, s_{2})  &< d_{\widehat{G}}(t_{12}, s_{1}) -  d_{\widehat{G}}(x, t_{12}) + d_{\widehat{G}}(x, t_{12}) + d_{\widehat{G}}(t_{12}, s_{2})\\
				& = d_{\widehat{G}}(t_{12},s_{1})+ d_{\widehat{G}}(t_{12}, s_{2}),
			\end{split}
		\end{equation*}
		as desired.
	\end{proof}
	We observe that \Cref{lm:bunch-disjiont2} follows directly from \Cref{clm:jk-12-dist} and \Cref{clm:jk-12-b2}.
\end{proof}

We now continue the proof of \Cref{thm:LW-undirected}.  Consider the subgraph $H = \cup_{1\leq i\leq h}\bunch(s_{i})$. We construct a $K_h$ minor of $H$ as follows (see \Cref{fig:VCGen}(f)). Let $\pi(s_{i},t_{ij},G)$ be a path in $\bunch(s_{i})$ such that $i< j$. If $t_{ij}$ is an endpoint of  $\bunch(s_{j})$, then we contract $\pi(s_{i},t_{ij},G)$  to $s_{i}$. Otherwise, we contract $\pi(s_{i},t_{ij},G)\setminus\{t_{ij}\}$ to $v_i$; \Cref{lm:bunch-disjiont2} implies that $t_{ij}$ will be contracted to $v_j$. The resulting graph is a $K_h$-minor of $H$ as the paths 	$\pi(s_{i},t_{ij},G)$ and $\pi(s_j,t_{ij},G)$ for any $i\not= j$ are internally disjoint  by \Cref{obs:disjoint}. This completes the proof of  \Cref{thm:LW-undirected}.

\begin{remark}\label{rm:LP} We remark the following regarding  \Cref{thm:LW-undirected}:
	\begin{itemize}
		\item The proof of \Cref{thm:LW-undirected} breaks down if we apply it to the set system by Li and Parter in \Cref{def:LP-setsystem}. Specifically, in \Cref{eq:a1234-sum2}, $d_G(v_{1j}, s_0) + d_G(v_{2\ell}, s_0)$ will be replaced by $d_G(v_{1j}, s_0) + d_G(v_{2\ell}, s_1)$ while in \Cref{eq:a1234-sum-final}, $d_G(v_{1j}, s_0) + d_G(v_{2\ell}, s_0)$ will be replaced by $d_G(v_{1j}, s_1) + d_G(v_{2\ell}, s_0)$, and hence we could not obtain a contradiction in the proof of \Cref{clm:jk-12-dist}. The same happens to the proof of \Cref{clm:jk-12-b2}. 
		
		\item  	The VC dimension bound obtained by Li and Parter~\cite{LP19}  is $3$ for the setting of $S$ on the outer face of a planar graph $G$, while our \Cref{thm:LW-undirected} gives VC dimension $4$. However, we can modify the proof slightly to improve the VC dimension to $3$ by  only requiring that $G$ excludes a $K_h$-minor where each vertex of the clique minor must correspond to a connected subgraph of $G$ containing at least one vertex in $S$. We say that $G$ is\emph{ $S$-restricted $K_h$-minor-free}. (A $K_h$-minor-free graph is $S$-restricted $K_h$-minor-free graph for any subset $S$.) The graph and the vertex set $S$ considered in the setting of Li and Parter is $S$-restricted $K_4$-minor-free and hence \Cref{thm:LW-undirected} gives VC dimension bound of $3$, matching the original bound of Li and Parter.
	\end{itemize}
\end{remark}

%% file: directed.tex
\section{VC Dimension of Digraphs and Applications}

In this section, $G = (V,E)$ denotes a $K_h$-minor-free \emph{digraphs}. $G$ could be weighted or unweighted. In bounding the VC-dimension, we allow edges of $G$ to have arbitrary non-negative weights, while in the algorithmic applications, $G$ is unweighted.

\input{LPDir}

\subsection{VC dimension of \texorpdfstring{$\vec{\mathcal{B}}(G)$}{BG}}\label{subsec:VC-dim-diballs}

We show \Cref{thm:diball-vc}, which states that the set system of balls $\lvec{B}(G)$ defined in \Cref{eq:diball} is a VC set system. We tailor the proof by Bousquet and Thomass{\'{e}}~\cite{BT15} for the undirected case to the directed case.

\DiBallVC*

The proof follows the presentation of the proof of \Cref{thm:LW-undirected} though several details are different. Specifically, we  assume for contradiction that $\lvec{B}(G)$ shatters a set $X = \{v_1,v_2,\ldots, v_h\}\subseteq V$ of size $h$.  Then for every $i\not=j$, there is a ball $\lvec{B}(t_{ij}, r_{ij})$ such that $\lvec{B}(t_{ij}, r_{ij})\cap X = \{v_i,v_j\}$. We choose $t_{ij}$ and $r_{ij}$ such that
\begin{equation}\label{eq:choice-tij}
	r_{ij} \mbox{ is minimum.}
\end{equation}
We then can assume that $r_{ij}  = \max\{d_{G}(t_{ij}\rightarrow v_i),d_{G}(t_{ij} \rightarrow v_j)\}$ as otherwise, $r_{ij} > \max\{d_{G}(t_{ij}\rightarrow v_i),d_{G}(t_{ij} \rightarrow v_j)\}$ and we can always set $r_{ij}$ to be $\max\{d_{G}(t_{ij}\rightarrow v_i),d_{G}(t_{ij} \rightarrow v_j)\}$. Our goal is to construct a $K_h$-minor of $G$ as we did in the proof of \Cref{thm:LW-undirected}.  We observe that \Cref{obs:disjoint} remains true in this setting of digraphs.

\begin{figure}[!htb]
	\centering{\includegraphics[width=.8\textwidth]{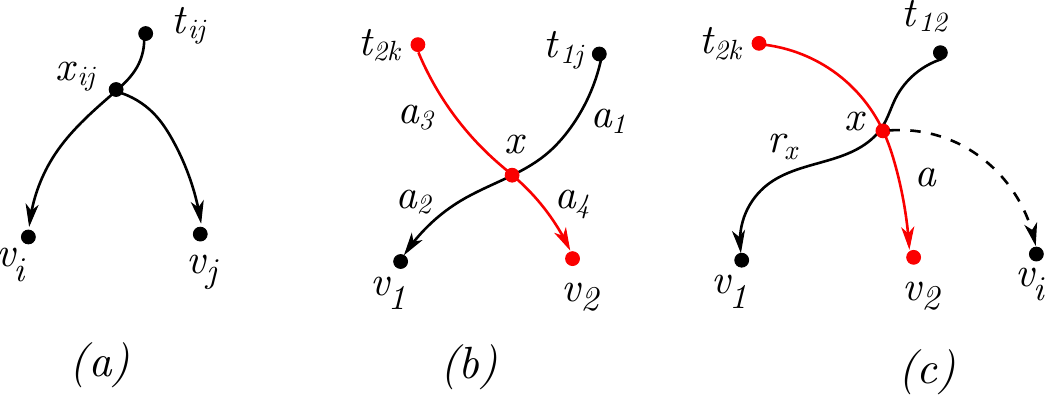}}
	\caption{Illustration for the proof of \Cref{thm:diball-vc}.}\label{fig:diballs}
\end{figure}

\begin{observation}\label{obs:disjont-di}	$\pi(t_{ij}\rightarrow v_i,G)$ and $\pi(t_{ij}\rightarrow v_j,G)$ are internally disjoint.
\end{observation}
\begin{proof}
	Suppose otherwise; there would be a vertex $x_{ij} \in (\pi(t_{ij}\rightarrow v_i,G)\cap \pi(t_{ij}\rightarrow v_j,G))\setminus \{t_{ij}\}$. Observe that  $\lvec{B}(x_{ij},r_{ij}-d_G(t_{ij} \rightarrow  x_{ij}))\cap X = \{v_i,v_j\}$, contradicting the choice of $r_{ij}$ in \Cref{eq:choice-tij}; see \Cref{fig:diballs}(a). 
\end{proof}

For each $v_i$, we define $\bunch(v_i)$ as in \Cref{eq:bunch2}, ignoring the directions of the paths. 
\begin{equation}\label{eq:bunch-dir}
	\bunch(v_i) = \cup_{j\not=i} \pi(t_{ij}\not\rightarrow v_i,G)
\end{equation}
Here $\pi(t_{ij}\not\rightarrow v_i,G)$ is an undirected path obtained by ignoring the directions of edges in $\pi(t_{ij}\rightarrow v_i,G)$. The proof of \Cref{thm:LW-undirected} in \Cref{subsec:Vc-dim-LPhat} implies that the existence of a $K_h$-minor is reduced to showing the following lemma.
\begin{lemma}\label{lm:di-bunch-disjiont} For every $a\not=b$, $\bunch(v_a)\cap \bunch(v_b) = \{t_{ab}\}$. Furthermore, $t_{ab}$ is either an endpoint of $\bunch(v_a)$, or an endpoint of $\bunch(v_b)$, or both.
\end{lemma}
\begin{proof}
	We follow the same proof strategy of \Cref{lm:bunch-disjiont2}: considering $a = 1$ and $b = 2$. 
	Let $\pi(t_{1j} \rightarrow v_1,G)$ and $\pi(t_{2k}\rightarrow v_2,G)$ be paths whose undirected counterparts are in $\bunch(v_1)$ and $\bunch(v_2)$, respectively.  The following claim is analogous to \Cref{clm:jk-12-dist}.
	
	\begin{claim}\label{clm:jk-12-a} If $\{j,k\}\cap\{1,2\} = \emptyset$, then $\pi(t_{1j}\rightarrow v_1,G)\cap  \pi(t_{2k}\rightarrow v_2,G) = \emptyset$.
	\end{claim}
	\begin{proof}
		Suppose otherwise, there exists $x\in \pi(t_{1j}\rightarrow v_1,G)\cap  \pi(t_{2k}\rightarrow v_2,G)$; see \Cref{fig:diballs}(b). Let:
		\begin{equation*}
			\begin{split}
				a_1 = d_G(t_{1j}\rightarrow x)&\qquad a_2 = d_G(x \rightarrow v_1)\\
				a_3 = d_G(t_{2k}\rightarrow x) &\qquad a_4 = d_G(x \rightarrow v_2)
			\end{split}
		\end{equation*}
		Since $v_1 \in \lvec{B}(t_{1j}, r_{1j})$ and $v_2 \in \lvec{B}(t_{2k}, r_{2k})$, 	$a_1 + a_2 \leq r_{1j}$ and $a_3+ a_4\leq r_{2k}$. This implies that 
		\begin{equation}\label{eq:a1234-sum1}
			a_1 + a_2 + a_3+ a_4\leq  r_{1j}  + r_{2k}
		\end{equation}
		On the other hand, $v_2\not\in  \lvec{B}(t_{1j}, r_{1j})$ and $v_1\not\in  \lvec{B}(t_{2k}, r_{2k})$, which gives	$a_1 + a_4 > r_{1j}$ and $a_2+ a_3 > r_{2k}$. This implies that $a_1 + a_2 + a_3+ a_4 >  r_{1j}  + r_{2k}$, contradicting \Cref{eq:a1234-sum1}. Thus, $x$ does not exist.
	\end{proof}
	
	The proof of the lemma follows directly from the following claim.
	\begin{claim}\label{clm:jk-12-b} If $\{j,k\}\cap\{1,2\}  \not=\emptyset$, then $\pi(t_{1j}\rightarrow v_1,G)\cap  \pi(t_{2k}\rightarrow v_2,G)\subseteq \{t_{12}\}$, and that $t_{12}$ is either an endpoint of $\bunch(v_1)$ or $\bunch(v_2)$ or both.
	\end{claim}
	\begin{proof} W.l.o.g., we assume that $j = 2$ and $k\not=1$.  Suppose that there exists $x\in \pi(t_{12}\rightarrow v_1,G)\cap  \pi(t_{2k}\rightarrow v_2,G)$ such that $x\not= t_{12}$. Let $a = d_G(x\rightarrow v_2)$. Then $d_G(x\rightarrow v_1) > a$ as otherwise, $v_1 \in \lvec{B}(t_{2k}, r_{2k})$, a contradiction. Let $r_x = d_G(x\rightarrow v_1)$. Then $\{v_2,v_1\}\subseteq \lvec{B}(x,r_x)\cap X $. We claim that $\lvec{B}(x,r_x)\cap X$ contains no other vertex other than  $v_1,v_2$; see \Cref{fig:diballs}(c).
		
		Suppose otherwise, there exists $v_i\in \lvec{B}(x,r_x)\cap X$ for $v_i \not= v_1,v_2$. Then $v_i \in  \lvec{B}(x,r_x)$  and hence $d_G(x\rightarrow v_i) \leq r_x = d_G(x\rightarrow v_1)$. This implies that $d_G(t_{12},v_i)\leq d_G(t_{12}\rightarrow v_1)\leq r_{12}$; that is, $v_i$ also belongs to the ball $\lvec{B}(t_{12},r_{12})$ contradicting the fact that $\lvec{B}(t_{12},r_{12})$  only shatters $\{v_1,v_2\}$.
		
		Since $\lvec{B}(x,r_x)\cap X$ contains no other vertex other than  $v_1,v_2$ and  $r_x < r_{12}$, we obtain a contradiction to the choice of $t_{12}$ in \Cref{eq:choice-tij}, as $\max\{d_{G}(x\rightarrow v_1),d_{G}(x \rightarrow v_2)\}  < \max\{d_{G}(t_{12}\rightarrow v_1),d_{G}(t_{12} \rightarrow v_2)\} $.
	\end{proof}
	The lemma then follows directly from \Cref{clm:jk-12-a} and \Cref{clm:jk-12-b}. 
\end{proof}

\subsection{Algorithmic Applications}

In this section, we explore algorithmic applications of two VC set systems $\vlp_{G,M}$ and $\lvec{B}(G)$.  \emph{Digraphs in this section are \emph{unweighted} and hence the distances are unweighted directed distances.} A central concept in the algorithmic applications of $\wlp_{G,M}$ in undirected graphs in \Cref{subsec:apps-undirected} is the notion of patterns and polynomial bounds on the number of patterns in a connected subgraph in \Cref{lm:pattern-bound-undir}. The same bound on the number of patterns \emph{completely breaks down} in digraphs, as the triangle inequality no longer holds. Only an asymmetric version fo the triangle inequality holds in digraphs, but this is not enough for deriving \Cref{lm:pattern-bound-undir} in digraphs. Indeed, we believe that  \Cref{lm:pattern-bound-undir} does not hold in digraphs. The implication of  not having a polynomial bound on the number of patterns is clear:  we could not easily derive analogous algorithmic results presented in \Cref{subsec:apps-undirected} for digraphs.  Instead, obtain similar results using  $\vlp_{G,M}$ and $\lvec{B}(G)$.

First, we devise a new way to exploit the set system of balls  $\lvec{B}(G)$ to design a distance oracle for digraphs with truly subquadratic space and logarithmic query time. The VC set system of balls is very hard to manipulate, as evidenced in the work of Ducoffe, Habib, and Viennot~\cite{DHV20} since it does not encode distances directly into the system. Thus, we believe that our technique is of independent interest; the details are in \Cref{subsec:oracle-digraphs}.

Second,  we modify the notion of patterns to include $\pm\infty$, called \emph{infinite patterns}, as a marker for the failure of the triangle inequality. We then are able to bound the number of infinite patterns, obtaining a lemma analogous to \Cref{lm:pattern-bound-undir}. We note that we still do not know how to exploit infinite patterns in constructing distance oracles in digraphs, as they do not enjoy the same properties as their (finite) counterpart. However, we are able to exploit infinite patterns to design truly subquadratic time algorithms for computing all-vertices eccentricities and the diameter of digraphs. The technical details are in \Cref{subsec:oracle-digraphs}.

\subsubsection{Distance oracle in digraphs.~}\label{subsec:oracle-digraphs} 

In this section, we construct an exact distance oracle for unweighted minor-free \emph{digraphs} with  $\tilde{O}(n^{2-\frac{1}{2(h-2)}})$ space and $O(\log(n))$ query time as described in \Cref{cor1:oracle-digraph}. We will use a well-known property of VC set system restricted to a subset, as described in the following lemma.




\begin{lemma}\label{lm:VC-restriction} Let $\mathcal{F}$ be a set system of a ground set $U$ of VC-dimension $d\geq 1$. Let $X$ be any subset of $U$. Then $\mathcal{F}_{X} = \{Y\cap X: Y\in \mathcal{F}\}$ has VC dimension at most $d$. We call $\mathcal{F}_{X}$ the $X$-restriction of $\mathcal{F}$. 
\end{lemma}

\begin{figure}[!htb]
	\centering{\includegraphics[width=.8\textwidth]{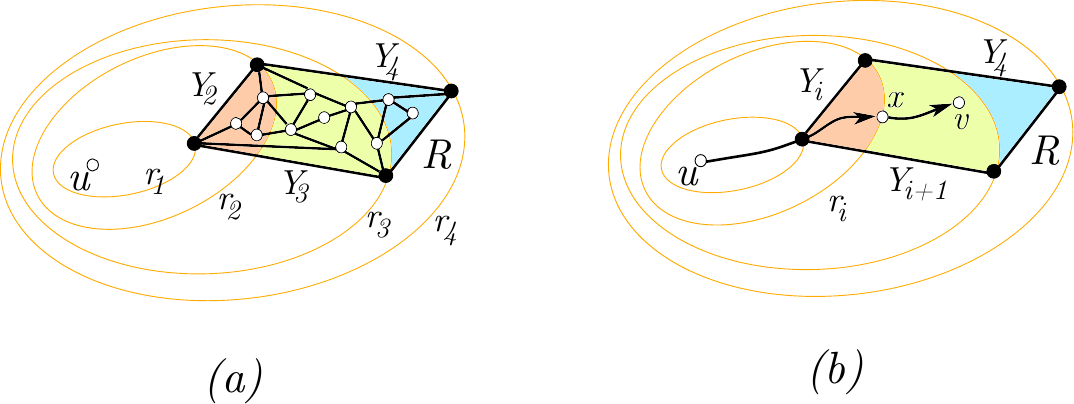}}
	\caption{(a) A region $R$ with 4 boundary vertices; the set $Y_1 = \vec{B}(u,r_1)\cap R$ only contains the boundary vertex closest to $u$. We ignore the directions of edges in $R$ in this figure for better visualization. (b) Querying distance from $u$ to $v$.}\label{fig:oracle}
\end{figure}

\paragraph{Construction.} The construction has three steps.

\begin{itemize}
	\item \textbf{(Step 1).~} Construct an $r$-division $\mathcal{R}$ of $G$ with $r = n^{2/(2h-1)}$, and for each cluster $R\in \mathcal{R}$, we store the exact distances of all pairs of vertices in $R$. Let $\lvec{\mathcal{B}}_R$ be the $V(R)$-restriction of $\lvec{\mathcal{B}}(G)$. We store (the IDs of) the sets of $\lvec{\mathcal{B}}_R$ in a table. 
	\item \textbf{(Step 2).~} For each cluster $R\in \mathcal{R}$ and each vertex $v\in R$, we store: (2a) the distance $d_G(s\rightarrow v)$ from each vertex $s\in \partial R$ to $v$; (2b) for each set $Y \in \lvec{\mathcal{B}}_R$, store  $d_G(Y\rightarrow v) \defi \min_{y\in Y} d_G(y\rightarrow v)$.
	
	\item \textbf{(Step 3).~} For each cluster $R\in \mathcal{R}$ and each vertex $u\not\in R$, let $k_R = |\partial R|$. Let $\lvec{B}(u,r_1),\ldots, \lvec{B}(u,r_{k_R})$ be a sequence of nest balls centered at $u$ where $r_1\leq r_2 \leq \ldots \leq r_{k_R}$ such that $\lvec{B}(u,r_1)$ is the smallest ball containing at least one vertex of $\partial R$, and $\lvec{B}(u,r_i)$ is the smallest ball containing at least one vertex of $\partial R\setminus \lvec{B}(u,r_{i-1})$; see \Cref{fig:oracle}(a). (The number of balls could be smaller than $k_R$; for simplicity, we assume that there are exactly $k_R$ balls.) Then we store at $u$ the radius $r_i$ and (the  IDs of) the restriction $Y_i = \lvec{B}(u,r_i)\cap V(R)$  for all $i\in [k_R]$ in a list $L(u,R)$. Note that $Y_i\in \lvec{\mathcal{B}}_R$ by the construction in (Step 1). We also store the distance $d_G(u\rightarrow s)$ from $u$ to every boundary vertex $s\in \partial R$. 
\end{itemize}

\paragraph{Querying distances.~} Given two vertices $u$ and $v$, if there is a cluster $R$ containing both $u$ and $v$, we can simply look up their distance stored at $R$ in $O(1)$ time. Otherwise, let $R$ be the cluster containing $v$. Let $Y_i = \lvec{B}(u,r_{i})\cap V(R)$. We then do a binary search on the list $L(u,R)$ to find the first radius $r_i$ such that $v\not\in Y_i$ and $v\in Y_{i+1}$; see \Cref{fig:oracle}(b). Note that we can check whether $v$ is in $Y_i$ or not in $O(1)$ time by the construction in (Step 2), in particular (2b) since $v\in Y_i$ if and only if $d_G(Y_i\rightarrow v) =0$. We then return:
\begin{equation}\label{eq:digraph-returned-dist}
	r_i + d_G(Y_i\rightarrow v)
\end{equation}
as the distance from $u$ to $v$. We note that $d_G(Y_i\rightarrow v)$ is stored in (2b) of (Step 2), so we can look up this distance in $O(1)$ time.

The query time is dominated by the time to do binary search on  $L(u,R)$, which is $O(\log |L(u,R)|) = O(\log|\partial R|) = O(\log r) = O(\log n)$.

\paragraph{Correctness.~} By the definition of $Y_{i}$ and $Y_{i+1}$, the shortest path from $u$ to $v$ must go through a vertex in $\partial R\cap Y_{i}$. Let $x$ be the last vertex on $\pi(u\rightarrow v, G)$ that is contained in $Y_i$; see \Cref{fig:oracle}(b). Then $d_G(u\rightarrow v) = d_G(u\rightarrow x) + d_G(x\rightarrow v)$. Since $G$ is unweighted, it must be that $d_G(u\rightarrow x) = r_i$. Furthermore,  $d_G(x\rightarrow v) = d_G(Y_i \rightarrow v)$ since otherwise, $d_G(x\rightarrow v) > d_G(Y_i,v)$ which means there is a path from $u$ to $v$ of length less than $d_G(u\rightarrow v)$, a contradiction. Thus, $d_G(u\rightarrow v) = r_i + d_G(Y_i \rightarrow v)$ as desired. 

\paragraph{Space analysis.~}  By \Cref{thm:diball-vc} and \Cref{lm:VC-restriction}, $\lvec{\mathcal{B}}_R$ has VC-dimension at most $h-1$. By \Cref{lm:SS}, $|\lvec{\mathcal{B}}_R| = O(r^{h-1})$ and hence the total space of Step 1 is $\tilde{O}((n/\sqrt{r})(r^{h-1}+r^2) = \tilde{O}(nr^{h-3/2}))$. The total space of Step 2 is $\tilde{O}((n/\sqrt{r})(r\cdot r^{h-1})) = \tilde{O}(nr^{h-3/2})$. For each vertex $u$ and cluster $R$ in Step 3, the total space is $O(|\partial R|)$. Thus, the total space of Step 3 is $n$ times the total number of boundary vertices, which is  $\tilde{O}(n^2/\sqrt{r})$ by \Cref{lm:r-division}. In summary, the total space of the oracle is:

\begin{equation}
	\tilde{O}(nr^{h-3/2} + n^2/\sqrt{r}) = \tilde{O}(n^{2-\frac{1}{2(h-2)}}) 
\end{equation}
when  $r = n^{1/(h-2)}$.

\subsubsection{Computing all-vertices eccentricities and diameter.}

\paragraph{Infinite patterns.} Let $H$ be an induced sub-digraph of $G$; $H$ might or might not be (even weakly) connected.  Recall that $\partial H$ is the set of all boundary vertices of $H$.  Let $r = |V(H)|$ and $b = |\partial H|$. Fix an arbitrary \emph{sequence} $\sigma_H$ of vertices of $\partial H$, which is a linear order of  $\partial H$. We write $\sigma_H = \langle s_0,s_1,\ldots, s_{b-1}\rangle$. For each vertex $v\in V$, we define an \emph{infinite pattern} of $v$ w.r.t $\sigma_H$, denoted by $\patt_v$ be a $b$-dimensional vector where for each $i \in [0, b-1]$

\begin{equation}\label{eq:infinite}
	\patt_v[i] =  \begin{cases}
		-\infty &\text{if $d_G(v\rightarrow s_i)- d_G(v\rightarrow s_0) \leq -r$}\\
		d_G(v\rightarrow s_i)- d_G(v\rightarrow s_0) &\text{if $-(r-1) \leq d_G(v\rightarrow s_i)- d_G(v\rightarrow s_0) \leq r-1$}\\
		+\infty &\text{if $d_G(v\rightarrow s_i)- d_G(v\rightarrow s_0) \geq r$}	
	\end{cases}
\end{equation}

In particular, two values $-\infty$ and $+\infty$ are used to mark that the distance $d_G(v\rightarrow s_i)$ is far smaller or larger than $d_G(v\rightarrow s_0)$. We have the following lemma analogous to \Cref{lm:pattern-bound-undir}.

\begin{lemma}\label{lm:pattern-bound-di} $H$ be an induced sub-digraph of a $K_h$-minor-free digraph $G$, and $\sigma_H$ be an arbitrary sequence of vertices in $\partial H$. Let $P = \{\patt_v: v\in V\}$ be the set of all infinite patterns w.r.t. $\sigma_H$. Then $|P| = O((|\partial H|\cdot |V(H)|)^{h^2})$. 
\end{lemma}
\begin{proof} The proof follows the same line of the proof of \Cref{lm:pattern-bound-undir}: we show that there is a bijection between the set of patterns $P$ and  $\wlp_{G,M}(\partial H)$ for an appropriate choice if $M$.  Let $M = \{-r, -(r-1), \ldots, (r-1), +r\}$. Observe that $|M| \leq 2r+1$. Consider the VC set system $\wlp_{G,M}(\partial H)$, which has VC dimension at most $h^2$ by \Cref{thm:LW-digraph}. By the Sauer–Shelah Lemma (\Cref{lm:SS}), we have  $|\wlp_{G,M}(\partial H)| = O((br)^{h^2})$. To see that there is a bijection between the set of patterns $P$ and  $\wlp_{G,M}(\partial H)$, we simply flatten each pattern $\patt_v$ to obtain a set $\widehat{X}_v \in \wlp_{G,M}(\partial H)$ in exactly the same way we did in \Cref{lm:pattern-bound-undir}.	
\end{proof}

We now define the distance from an infinite pattern to a vertex.  Let $v$ be a vertex in $H$, and $\patt$ be a pattern (of some vertex $u$) w.r.t. $\sigma_H$. Let $\reach(v,\partial H)$ be the set of boundary vertices of $H$ that can reach $v$ (via directed paths) \emph{in $H$}. (We do not count boundary vertices that can reach $v$ in $G$.)  We define the distance from $\patt$ to $v$, denoted by $d(\patt\rightarrow v)$, to be:
\begin{equation}\label{eq:dis-patt-dir}
	d(\patt \rightarrow v) =  \begin{cases}
		\mathrm{undefined} &\text{if $\patt[i] = +\infty$ for some $i$}\\
		\min_{s_i \in \reach(v,\partial H)}\{d_G(s_i\rightarrow v) + \patt[i]\} &\text{otherwise}	
	\end{cases}
\end{equation}

In undirected graphs, we show in \Cref{lm:dist-via-pattern} that if $u\not\in V(H)$ and $v\in V(H)$, then $d_G(u,s_0) + d_G(\patt_u,v) = d_G(u,v)$ where $\patt_u$ is the pattern of $u$. This no longer holds in digraphs. In particular, $d_G(u \rightarrow s_0) + d_G(\patt_u \rightarrow v)$ now may be undefined or larger than $d_G(u\rightarrow v)$. However, we are still able to extract information by looking at all distances $\{d_G(\patt_u \rightarrow v)\}_{v\in V(H)}$.  In particular, we show in the following lemma that we can recover \emph{the maximum distance} from $u$ to a vertex in $H$, via  $\{d_G(\patt_u \rightarrow v)\}_{v\in V(H)}$, \emph{provided that $d_G(u\rightarrow s_0)$ is the maximum among all boundary vertices}.

\begin{lemma}\label{lm:dist-dir-via-pattern} Let $u \in V\setminus V(H)$ be a vertex not in $H$, and $\patt_u$ be the pattern of $u$ w.r.t $\sigma_H$.  Define:
	\begin{equation}
		\Delta(u\rightarrow H)= d_G(u\rightarrow s_0) + \max_{v\in V(H)}\{ d(\patt_u\rightarrow v)\}
	\end{equation}
\end{lemma}
If $d_G(u\rightarrow s_0) = \max_{0\leq i \leq |\partial H|-1}\{d_G(u\rightarrow s_i)\}$,  then $\Delta(u\rightarrow H) = \max_{v\in V(H)} d_G(u\rightarrow v)$.
\begin{proof}
	As $d_G(u\rightarrow s_0)$ is maximum, $d_G(u,s_i) - d_G(u\rightarrow s_0)\leq 0$ for every $i\in [0,b-1]$. Thus, no entry of $\patt_u$ is $+\infty$. Therefore, $d(\patt_u\rightarrow v)$ is defined (but could still be $-\infty$). 
	
	\begin{claim}\label{clm:minus-infty} If there exists a boundary vertex $s_i \in \reach(v,\partial H)$ such that $\patt_u[i] = -\infty$, then $d_G(u\rightarrow v) < d_G(u\rightarrow s_0)$.
	\end{claim} 
	\begin{proof}$\patt_u[i] = -\infty$ implies that $d_G(u,s_i) \leq d_G(u\rightarrow s_0) - r$. As $v$ is reachable from $s_i$ in $H$, there is a path of length at most $|V(H)|-1= r-1$ from $s_i$ to $v$, meaning that $d_G(s_i\rightarrow v) \leq r-1$. Thus, $d_G(u\rightarrow v)\leq d_G(u\rightarrow s_i) + d_G(s_i\rightarrow v)\leq d_G(u\rightarrow s_0) - r + (r-1) < d_G(u\rightarrow s_0)$ as claimed.
	\end{proof}
	
	By definition of the distance in \Cref{eq:dis-patt-dir}, if there exists a boundary vertex $s_i \in \reach(v,\partial H)$ such that $\patt_u[i] = -\infty$, then $d(\patt_u\rightarrow v) = -\infty$ and hence would have no effect in the computation of $\Delta(u\rightarrow H)$. And by \Cref{clm:minus-infty}, such a vertex $v$ also do not contribute to $ \max_{v\in V(H)} d_G(u\rightarrow v)$. Thus, we only need to consider vertices $v$ such that for every boundary vertex $s_i \in \reach(v,\partial H)$, $\patt_u[i] \not=-\infty$. We claim that for such vertices, $d_G(u\rightarrow s_0) + d(\patt_u\rightarrow v)$ is the distance from $u$ to $v$.
	
	\begin{claim}\label{clm:real-dist} If for every boundary vertex $s_i \in \reach(v,\partial H)$, $\patt_u[i] \not= -\infty$, then $d_G(u\rightarrow v) = d_G(u\rightarrow s_0) +  d(\patt_u\rightarrow v)$.
	\end{claim}
	\begin{proof} The assumption of the claim implies that $d_G(u\rightarrow s_{i}) = d_G(u\rightarrow s_0) + \patt_u[i]$ for every boundary vertex $s_i \in \reach(v,\partial H)$. Let $s_\ell$ for some $\ell\in [0,b-1]$ be the boundary vertex  on the path $\pi(u\rightarrow v, G)$ \emph{furthest} from $u$. That is, the subpath from $s_{\ell}$ to $v$ of  $\pi(u\rightarrow v, G)$ lies entirely in $H$. Thus, $s_{\ell}\in \reach(v,\partial H)$ and $d_G(s_{\ell}\rightarrow v) = d_H(s_{\ell}\rightarrow v)$. Then: 
		\begin{equation*}
			\begin{split}
				d_G(u \rightarrow v) &= d_G(u \rightarrow s_\ell) + d_G(s_\ell \rightarrow v)\\ 
				&= \min_{s_i\in \reach(v,\partial H)}\{d_G(u\rightarrow s_{i}) + d_G(s_{i} \rightarrow v)\}\\
				&=   \min_{s_i\in \reach(v,\partial H)}\{ d_G(u\rightarrow s_0) + \patt_u[i] + d_G(s_{i} \rightarrow v)\}\\
				&= d_G(u\rightarrow s_0) + \min_{s_i\in \reach(v,\partial H)}\{  d_G(v\rightarrow s_i) + \patt_u[i]\}\\ & = d_G(u\rightarrow s_0) + d_G(\patt_u\rightarrow v)~,
			\end{split}
		\end{equation*}
		as desired.
	\end{proof}

	Let $\delta = \max_{v\in V(H)} d_G(u\rightarrow v)$ and $v^*\in V(H)$ be such that $d_G(u\rightarrow v^*) = \delta$. Observe that $\delta \geq d_G(u\rightarrow s_0)$ since $s_0$ is an eligible choice for $v^*$. We now show that for every boundary vertex $s_j$ such that $s_j\in \reach(v^*,\partial H)$, $\patt_u[j] \not= -\infty$. If so, by \Cref{clm:real-dist},  $\delta = d_G(u\rightarrow s_0) +  d(\patt_u\rightarrow v^*)$, which implies the lemma.

	To see that $\patt_u[j] \not= -\infty$, first observe that $d_G(s_j\rightarrow v^*)\leq r-1$ as $s_j\in \reach(v,\partial H)$, and that $d_G(u \rightarrow s_j) + d_G(s_j \rightarrow v^*) \geq d_G(u \rightarrow v^*) = \delta$. Thus, we have:
	\begin{equation*}
		d_G(u \rightarrow s_j) \geq \delta -  d_G(s_j \rightarrow v^*) \geq \delta - (r-1)\geq d_G(u\rightarrow s_0)- (r-1)
	\end{equation*}
	which gives $d_G(u \rightarrow s_j)  -d_G(u\rightarrow s_0) \geq - (r-1)$. Furthermore, by definition of $s_0$, $d_G(u \rightarrow s_j)  -d_G(u\rightarrow s_0) \leq 0$. Thus, $\patt_u[j] \not= -\infty$ as desired. 
\end{proof}

We call the first boundary vertex $s_0$ in a sequence of boundary vertex $\sigma_H$ of $H$ the \emph{base} of $\sigma_H$. We remark that in \Cref{lm:dist-dir-via-pattern}, it is important that the distance from $u$ to the base vertex satisfies $d_G(u\rightarrow s_0) = \max_{0\leq i \leq |\partial H|-1}\{d_G(u\rightarrow s_i)\}$, we call this condition the \emph{maximum base condition}. In general, for any fixed sequence $\sigma_H$, if the maximum base condition is satisfied for $u$, it might not be satisfied for some vertex $v$. Thus, in the following algorithm for computing all-vertices eccentricities, we have to consider $|\partial H|$ different boundary sequences, each has a different boundary vertex as the base. We note that only the base vertex is important; the order of remaining vertices in a sequence $\sigma_H$ could be arbitrary. 

\paragraph{The algorithm.~} The algorithm for computing all-vertices eccentricities has 3 steps. Here we focus on presenting the ideas and then discuss the implementation later. 

\begin{itemize}
	\item \textbf{(Step 1).~} Construct an $r$-division $\mathcal{R}$ of $G$ for $r = n^{2/(3h^2+6)}$. For each cluster $R\in \mathcal{R}$, we construct a set,  denoted by $\Gamma_R$, of $|\partial R|$ different sequences of boundary vertices of $R$ such that each sequence in $\Gamma_R$ admits a different boundary vertex as the base. We write $\Gamma_R = \{\sigma^{1}_R, \sigma^{2}_R, \ldots, \sigma^{|\partial R|}_R\}$. Then for each sequence $\sigma^t_{R}$ for $t\in [|\partial R|]$, we construct a set of infinite patterns w.r.t  $\sigma^t_R$: $P^t_R = \{u\in V: \patt^t_u\}$ where $\patt^t_u$ is the infinite pattern of $u$ w.r.t $\sigma^t_R$. Let $\mathcal{P}_R = \{P^t_R\}_t$.  
	
	
	\item \textbf{(Step 2).~} For each cluster $R\in \mathcal{R}$, 
	each pattern $\patt \in (\bigcup_{P_R\in \mathcal{P}_R}P_R)$, find $v = \argmax_{\tilde{v}\in V(R)} d(\patt\rightarrow \tilde{v})$; we exclude undefined distances in the search for $v$. That is, $v$ is the vertex that has maximum distance from $\patt$ over all vertices in $V(R)$; we say that $v$ is the \emph{furthest vertex} from $\patt$. We then store the distance $d(\patt \rightarrow v)$ in a table. 
	
	\item \textbf{(Step 3).~} We now compute $\ecc(u)$ for each vertex $u\in V$. For each cluster $R\in \mathcal{R}$, we compute the distance from $u$ to a vertex $v\in R$ furthest from $u$, denoted by $\Delta(u\rightarrow R)$, as follows.  Let $s_t$ be the furthest boundary vertex in $R$: $d_G(u\rightarrow s_t) = \max_{s\in \partial R} d_G(u\rightarrow s)$. Let $P^t_R$ be the set of infinite patterns w.r.t the boundary sequence that has $s_t$ as the base computed in (Step 1). Let $\patt^t_u$ be the pattern of $u$ in $P^t_R$. If $u\not\in R$, let $v$ be the furthest vertex from $\patt^t_u$, computed in (Step 2). Then  we return $\Delta(u\rightarrow R) = d_G(u \rightarrow s_t) + d(\patt^t_u \rightarrow v)$. Otherwise, for every vertex $v\in R$, we compute $\tilde{d}_G(u\rightarrow v) = \min\{d_G(u,s_0) + d(\patt^t_u\rightarrow v), d_{R}(u\rightarrow v)\}$ and finally return $\Delta(u,R) = \max_{v\in R}{\tilde{d}_G(u\rightarrow v)}$.
	\end{itemize}

As discussed in (Step 3), \Cref{lm:dist-dir-via-pattern} implies that the computed value $\ecc(u)$ is the eccentricity of $u$. We now show an efficient implementation and analyze its running time. 

\paragraph{Efficient implementation.~} Implementing the algorithm for digraphs shown above in truly subquadratic time turns out harder than the algorithm for undirected graphs in \Cref{subsec:diameter-unweighted}. One reason is that each vertex $u$ now is associated with up to $\sqrt{r}$ different pattern vectors, each for one boundary sequence, in the same cluster $R$. As each pattern vector has size up to $\sqrt{r}$, the total amount of information per vertex $u$, and per cluster $R$ is $O(r)$. The number of clusters is $\tilde{O}(n/\sqrt{r})$  in \Cref{lm:r-division}. The number of clusters can indeed be improved to $\tilde{O}(n/r)$ if one is willing to pay more running time. Even in the best case on the size of the number of clusters,  the total amount of computation, if done carelessly, is $\tilde{O}(n r\cdot (n/r) \cdot r) = \tilde{O}(n^2)$, which is larger than permitted.

The key idea in the implementation is not to compute all the patterns of all vertices in the graph. As we see in (Step 3), we only need to compute a pattern $\patt_u$ associated with a specific boundary sequence where the base of the sequence is the furthest boundary vertex; other boundary sequences are not relevant to compute $\Delta(u\rightarrow R)$. And this is what we will do: we \emph{will not} compute all the sets $\Gamma_R$ and $\mathcal{P}_R$ as described in (Step 1) upfront. Instead, we will implement (Step 3) directly first and then add patterns to $\mathcal{P}_R$ along the way we examine each vertex $u$.

Recall that  $B$ is the set of boundary vertices of the $r$-division $\mathcal{R}$: $B = \cup_{R\in \mathcal{R}} \partial R$.   Let $D(B\rightarrow V) = \{d_G(s\rightarrow v): (s,v)\in B\times V \}$,  $D(V\rightarrow B) = \{d_G(v\rightarrow s): (v, s)\in V\times B \}$. \Cref{obs:Tcal-Time} remains true here:

\begin{observation}\label{obs:Tcal-Dir-Time}$D(B\rightarrow V)$  and $D(V\rightarrow B)$ can be computed in time $\tilde{O}(n^{2}/\sqrt{r})$.
\end{observation}

We also obtain a polynomial bound on the number of infinite patterns as a corollary of \Cref{lm:pattern-bound-di}.

\begin{corollary}\label{cor:dir-pattern-count}$\sum_{P_R\in \mathcal{P}_R}|P_R| =  \tilde{O}(r^{(3h^2+1)/2})$ for every $R\in \mathcal{R}$.
\end{corollary}
\begin{proof}
	The number of infinite patterns per boundary sequence $\sigma_R$  is $O((|\partial R|\cdot |V(R)|)^{h^2}) = \tilde{O}(r^{3h^2/2})$. The corollary follows from the fact that we have up to $\sqrt{r}$ different boundary sequences.
\end{proof}

Now we show the detailed implementation of the algorithm, given $D(B\rightarrow V)$  and $D(V\rightarrow B)$. In (Step 1), we now only form all $|\partial R|$ boundary sequences -- the set $\Gamma_R$-- for each cluster $R$. The total running time per region is $O(|\partial R|^2) = \tilde{O}(r)$. By \Cref{lm:r-division}, the running time to find all $\{\Gamma_R\}_{R\in \mathcal{R}}$ is:
\begin{equation}\label{eq:Gamma-R}
	\tilde{O}(n \cdot r/\sqrt{r}) = \tilde{O}(n\cdot r^{1/2})
\end{equation}

Now we jump to (Step 3). For each vertex $u$ and each cluster $R$, we first find the furthest boundary vertex $s_t$, in $O(|\partial R|)$ time by looking through all the distances from $u$ to vertices $\partial R$ stored in  $D(V\rightarrow B)$. Thus, the total running time of finding all furthest boundary vertices over all vertices and all clusters is:
\begin{equation}\label{eq:furhtest-btime}
	n \sum_{R\in \mathcal{R}} |\partial R| = \tilde{O}(n^2/\sqrt{r})
\end{equation}
by \Cref{lm:r-division}. Now we know the boundary sequence $\sigma^t_R$, computed in (Step 1), as we know the furthest vertex $s_t$. We can compute $\patt^t_u$, which is the pattern satisfying the maximum base condition, in $O(|\partial R|)$ time; the same time it takes to find $s_t$. Thus, the total running time to find all infinite patterns satisfying the maximum base condition over all vertices and clusters is $\tilde{O}(n^2/\sqrt{r})$ by \Cref{eq:furhtest-btime}. 

Finally, we have to compute  $\max_{\tilde{v}\in R} d(\patt^t_u \rightarrow \tilde{v})$, and find the vertex $v$ which is furthest from $\patt^t_u$ in $R$. We could not naively iterate over all vertices in $R$ every time we examine a vertex $u$. Recall that the number of distinct infinite patterns per cluster $R$ is $\tilde{O}(r^{(3h^2+1)/2})$, and hence many vertices will share the same infinite patterns. We then could store results computed before in the table and do a table look up if we encounter the same pattern again. More specifically, we store a trie data structure $\mathcal{L}$: for a vertex $u$, if $\patt^t_u$ is not in $\mathcal{L}$, which we can check in $O(|\partial R|)$ time, then we iterate over all vertices in $R$ to find $v$, and store $v$  and $d(\patt^t_u \rightarrow v)$ in $\mathcal{L}$, keyed by $\patt^t_u$. Otherwise, we simply lookup $v$  and $d(\patt^t_u \rightarrow v)$ from $\mathcal{L}$. Modulo the running time to find $v$  and $d(\patt^t_u \rightarrow v)$ when $\patt^t_u \not\in \mathcal{L}$, the total time to look up  $v$  and $d(\patt^t_u \rightarrow v)$ will be $|\partial R|$, which is also the time to find $s_t$. Thus, the total running time is  $ \tilde{O}(n^2/\sqrt{r})$ by \Cref{eq:furhtest-btime}.

We now bound the running time to find  $v$  and $d(\patt \rightarrow v)$ for  every pattern $\patt \in (\bigcup_{P_R\in \mathcal{P}_R}P_R)$. For each given pattern $\patt$, computing the distance from $\patt$ to a vertex $\tilde{v} \in R$ can be done in $O(|\partial R|) = \tilde{O}(\sqrt{r})$ time by definition in \Cref{eq:dis-patt-dir}. Then finding $ \max_{v\in R} d(\patt\rightarrow v)$ can be done in $|V(R)| \tilde{O}(\sqrt{r}) = \tilde{O}(r^{3/2})$ time. Over all patterns in $ (\bigcup_{P_R\in \mathcal{P}_R}P_R)$, by \Cref{cor:dir-pattern-count}, the total running time is $\tilde{O}(r^{(3h^2+1)/2} \cdot r^{3/2}) = \tilde{O}(r^{(3h^2+6)/2})$. Over all clusters in $\mathcal{R}$, the running time is:

\begin{equation}\label{eq:patter-v-all}
	\tilde{O}(n/\sqrt{r}) \cdot \tilde{O}(r^{(3h^2+6)/2}) = \tilde{O}(nr^{(3h^2+5)/2})
\end{equation} 

In summary,  by \Cref{eq:Gamma-R}, \Cref{eq:furhtest-btime}, and \Cref{eq:patter-v-all}, the total running time to compute all-vertices eccentricities is:

\begin{equation}\label{eq:final-dir}
	\tilde{O}(nr^{(3h^2+5)/2}) + \tilde{O}(n^2/\sqrt{r}) = \tilde{O}(n^{2-1/(3h^2+6)})
\end{equation} 

when $r = n^{2/(3h^2+6)}$, as claimed in \Cref{cor:dimater-digraph}. This is also the running time to compute the directed diameter of $G$.

%% file: LPDir.tex
\subsection{VC dimension of \texorpdfstring{$\vec{\lp}_{G,M}$}{DiLP}}\label{subsec:VC-dim-vecLP}

In this section, we prove \Cref{thm:LW-digraph}, which we restate below.

\LWDir*

Suppose that $\vlp_{G,M}$ shatters a set $X = \{(s_{1},\Delta_1), (s_{1},\Delta_2)\ldots, (s_{q},\Delta_q)\}$ of size $q$. Our goal is to show that (the undirected counterpart of) $G$ has a clique minor of size at least $\lfloor \sqrt{q} \rfloor$, which gives the bound on the VC dimension of $\vlp_{G,M}$, as $\lfloor \sqrt{q} \rfloor\leq h-1$. The major difficulty in the proof is that, in digraphs, we do not have strong properties of $\bunch(s_{i})$---we construct $\bunch(s_{i})$ in the same way---as we do in the proof of \Cref{thm:LW-undirected} in \Cref{subsec:Vc-dim-LPhat}. More precisely, \Cref{lm:bunch-disjiont2} no longer holds. This makes the construction of the clique minor more difficult, and as a result, we could not show the linear bound on the VC dimension. On the other hand, we show that an analog of \Cref{obs:disjoint} suffices for our construction of a clique minor of size $\sqrt{q}$.

We now present the proof. By the same reasoning in \Cref{obs:s-diff}, we have that $s_{i}\not=s_{j}$ for all $i\not= j$. For every pair $(s_{i},\Delta_j), (s_{j},\Delta_j)$, let $t_{ij}$ be such that $\lvec{X}_{t_{ij}} \cap (V\times M) = \{(s_{i},\Delta_j), (s_{j},\Delta_j)\}$.

\begin{figure}[!htb]
	\centering{\includegraphics[width=.8\textwidth]{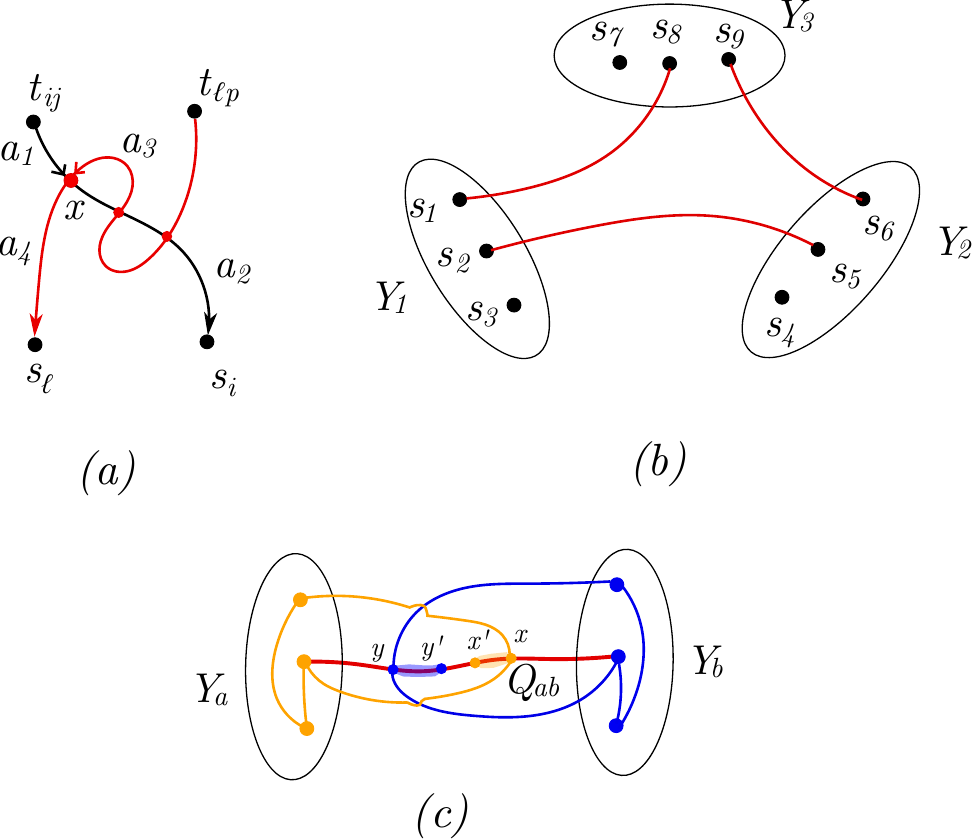}}
	\caption{(a) Assume that $x\in \pi(t_{ij}\rightarrow s_{i},G)\cap  \pi(t_{\ell p}\rightarrow s_{\ell},G)$; (b) three sets $Y_1,Y_2,Y_3$ when $k=9$ and three paths $Q_{12},Q_{23},Q_{13}$ where no two paths share the same endpoint; (c) $Q_{ab}[x,x']$ is added to $H_a$ (orange) and $Q_{ab}[y,y']$ is added to $H_b$ (blue).}\label{fig:directed-VC}
\end{figure}

\begin{lemma}\label{lm:LP-disjoint} If $i,j, \ell,p \in [q]$ are pairwise different, then $\pi(t_{ij}\rightarrow s_{i}, G)\cap \pi(t_{\ell p}\rightarrow s_{\ell}, G)=\emptyset$.
\end{lemma}
\begin{proof}
	Suppose otherwise, there exists $x\in \pi(t_{ij}\rightarrow s_{i}, G)\cap \pi(t_{\ell p}\rightarrow s_{\ell}, G)$; see \Cref{fig:directed-VC}(a). Let:
	\begin{equation*}
		\begin{split}
			a_1 = d_G(t_{ij}\rightarrow x)&\qquad a_2 = d_G(x\rightarrow s_{i})\\
			a_3 = d_G(t_{\ell p}\rightarrow x)&\qquad a_4 = d_G(x\rightarrow s_{\ell})
		\end{split}
	\end{equation*}
	Then 	$a_1 + a_2 \leq d_G(t_{ij}\rightarrow s_0) + \Delta_i$ and $a_3+ a_4\leq  d_G(t_{\ell p}\rightarrow s_0) + \Delta_{\ell}$. Thus, we have:
	\begin{equation*}
		a_1 + a_2 + a_3+ a_4\leq  d_G(t_{ij}\rightarrow s_0) + d_G(t_{\ell p}\rightarrow s_0)  + \Delta_i + \Delta_{\ell}
	\end{equation*}
	Furthermore, since $(s_{\ell},\Delta_{\ell})\not\in X_{t_{ij}}$, we have $d_G(t_{ij}\rightarrow s_{\ell}) > d_G(t_{ij}\rightarrow s_0)+\Delta_{\ell}$. This implies that  $a_1+ a_4 > d_G(t_{ij}\rightarrow s_0)+\Delta_{\ell}$. By the same argument, $a_2+a_3 > d_G(t_{\ell p}\rightarrow s_0)+\Delta_i$. Thus, $	a_1 + a_2 + a_3+ a_4>  d_G(t_{ij}\rightarrow s_0) + d_G(t_{\ell p}\rightarrow s_0)  + \Delta_i + \Delta_{\ell}$, a contradiction.
\end{proof}

We now ignore the direction of $G$ and focus on constructing a clique minor of size $\lfloor \sqrt{q} \rfloor$. 
Let $\pi(t_{ij}\not\rightarrow s_i,G)$ be the undirected path obtained by ignoring the direction of edges in $\pi(t_{ij}\rightarrow s_i,G)$. For every $i\not=j$ we denote by $P_{ij}$ the path from $s_i$ to $s_j$ obtained by simplifying the (undirected) walk from $s_i$ to $s_j$ obtained by gluing two paths $\pi(t_{ij} \not\rightarrow s_i,G)$ and $\pi(t_{ij} \not\rightarrow s_j,G)$ at $t_{ij}$. \Cref{lm:LP-disjoint} implies:

\begin{corollary}\label{cor:disjoint} $P_{ij}\cap P_{\ell p}  = \emptyset$ when $i,j,\ell,p$ are pairwise different.
\end{corollary} 

That is the two paths between two pairs of vertices in $X$ can intersect if and only if they share one endpoint. In this case, they could intersect in an arbitrarily complicated way.

We partition $X$ into $\sqrt{q}$ subsets $Y_1,\ldots, Y_{\sqrt{q}}$ each contains $\sqrt{q}$ vertices in $X$; for ease of notation, we assume that $\sqrt{q}$ is an integer. For every pair $(Y_a,Y_b)$ for $a,b\in [\sqrt{q}], a\not= b$, let $\mathcal{P}_{ab} = \{P_{ij}: s_i\in Y_a, s_j \in Y_b\}$ be the set of paths between $Y_a$ and $Y_b$. We then choose a path $Q_{ab}\in \mathcal{P}_{ab}$ such that the set of chosen paths, denoted by $\mathcal{Q} = \{Q_{ab}\}_{(a,b)\in [\sqrt{q}]\times [\sqrt{q}], a < b}$, has no two paths sharing the same endpoint; we can pick $\mathcal{Q}$ in a greedy manner. $\mathcal{Q}$ exists since each $Y_a$ has $\sqrt{q}$ vertices while we only need $\sqrt{q}-1$ paths in $\mathcal{Q}$ to connect $Y_a$ to other sets.  See \Cref{fig:directed-VC}(b).

We now construct a $K_{\sqrt{q}}$-minor as follows. For each $Y_a$, $a\in [\sqrt{q}]$, let $H_a = \cup_{s_i,s_j\in Y_a} P_{ij}$. Clearly, $H_a$ is connected and furthermore, by \Cref{cor:disjoint}, $V(H_a)\cap V(H_b) =\emptyset$.  Between $H_a$ and $H_b$, we have a path $Q_{ab} \in \mathcal{Q}$ that is vertex disjoint from all other paths in $\mathcal{Q}$. ($H_a$ and $H_b$ could contain vertices of $Q_{ab}$ other than its endpoints.) Since $V(H_a)\cap V(H_b) = \emptyset$, there must be a subpath $Q_{ab}[x,y]$ from a vertex $x$ to  a vertex $y$ such that $x\in H_a$ and $y\in H_b$ and no other vertex in $Q_{ab}[x,y]\setminus \{x,y\}$ belongs to $H_a\cup H_b$. (It could be that $Q_{ab}[xy]$ is an edge.) Pick an arbitrary edge $e_{ab} = (x',y')\in Q[x,y]$; we assume w.l.o.g that $x' \in Q_{ab}[x,y']$. Then we add $Q_{ab}[x,x']$ to $H_a$ and $Q_{ab}[y',y]$ to $H_b$. See Figure~\ref{fig:directed-VC}(c). Let $H'_a$ be the graph $H_a$ after applying this process to all pairs $(a,b)\in [\sqrt{q}]\times[\sqrt{q}], a < b$. Then $\{H'_a\}_{a\in [\sqrt{q}]}$  are pairwise vertex-disjoint, and there is an edge connecting every pair of graphs. These graphs induce a $K_{\sqrt{q}}$ of $G$, as desired.

%% file: lowerbound.tex
\section{Lower Bound for Directed VC-dim Edge Set System}

In this section, we prove the lower bound in \Cref{thm:diEdge-vc}, which we restate below.

\LBEdge*

We will construct a graph with a set of $r$ directed edges $X = \{e_1,e_2,\ldots, e_r\}$  on a path $P$ such that for every subset $Y\subseteq X$, there exists a vertex $v \not\in P$ such that $Y$ belongs to the shortest path tree rooted at $v$.

\begin{figure}[!htb]
	\centering{\includegraphics[width=.8\textwidth]{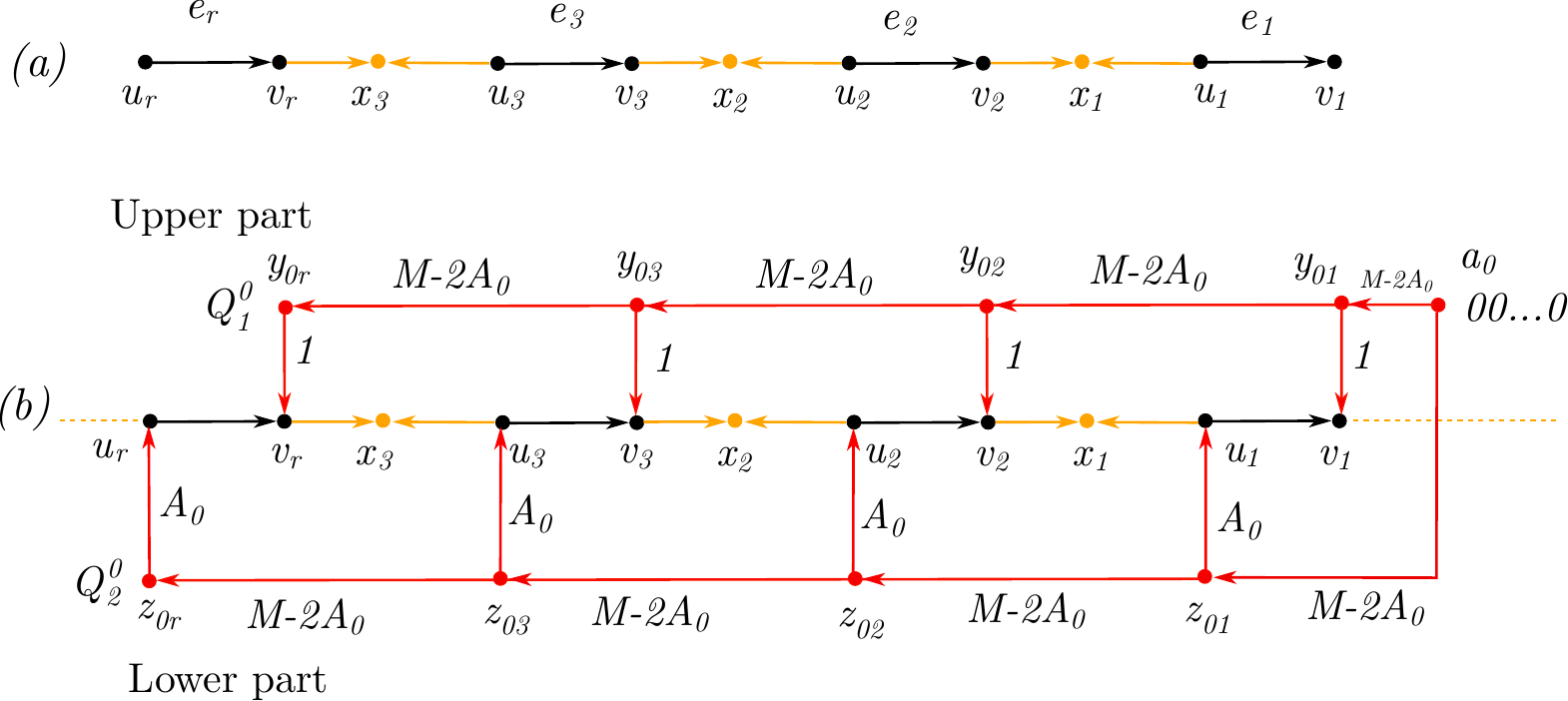}}
	\caption{(a)  the path $P$. (b) Vertex $a_0$, corresponding to the bit string $s_0 = 0 \ldots 0$ containing $r$ bits of 0s, connected to $P$ via directed edges.  
	}\label{fig:pathP}
\end{figure}

We construct $P$ as follows (see \Cref{fig:pathP}(a)). First, form the set $X$ of $r$ directed edges, where $e_i = (u_i \rightarrow v_i)$. Then add a path of length 2  between $v_{i+1}$ and $u_{i}$ consisting of two edges in different directions: $(v_{i+1}\rightarrow x_i), (u_i\rightarrow x_i)$. The idea is to ensure that no endpoint of $e_i$ can reach (or be reached by) other endpoints of other edges by going along the path $P$. Set the weight of each edge to be $1$.

Now we construct shortest path trees where each tree realizes a subset $Y$ of $X$. By realizing $Y$, we mean the subset $Y$ will be included in some shortest path tree, while other edges in $X\setminus Y$ will not be included by the same tree. We will add edges with \emph{integer weights} and finally we can turn them into unweighted edges by subdividing them.

Choose a sufficiently large number $M$ and $2^{r}$ other numbers: $1 \ll A_0 \leq A_1 \leq \ldots \leq A_{2^r-1}$ where $A_i = 4A_{i-1}$. $M$ will be sufficiently larger than all $\{A_i\}$. The following inequalities will be helpful:

\begin{equation} \label{eq:Ai}
	A_i \geq 2A_{t} + (A_{i-1} + A_{i-2} + \ldots + A_{0} + 1) \quad \mbox{ for any }t\leq i-1 
\end{equation}

Now consider all $2^{r}$ bit strings $\{0,1\}^r$, the $i$-th string, denoted by $s_i$, is the binary representation of $i$ for $i\in [0,2^{r}-1]$.  Starting from $s_0$, for each string $s_i$, we will add a new vertex $a_i$ to the graph along with some other vertices and directed edges. Some directed edges will be given weights based on $M$ and $A_i$. Vertex $a_i$ will be embedded outside the outer face of $G_{i-1}$, the directed graph constructed after step $i-1$. The final graph is $G = G_{2^{r}-1}$.

\begin{figure}[!htb]
	\centering{\includegraphics[width=.8\textwidth]{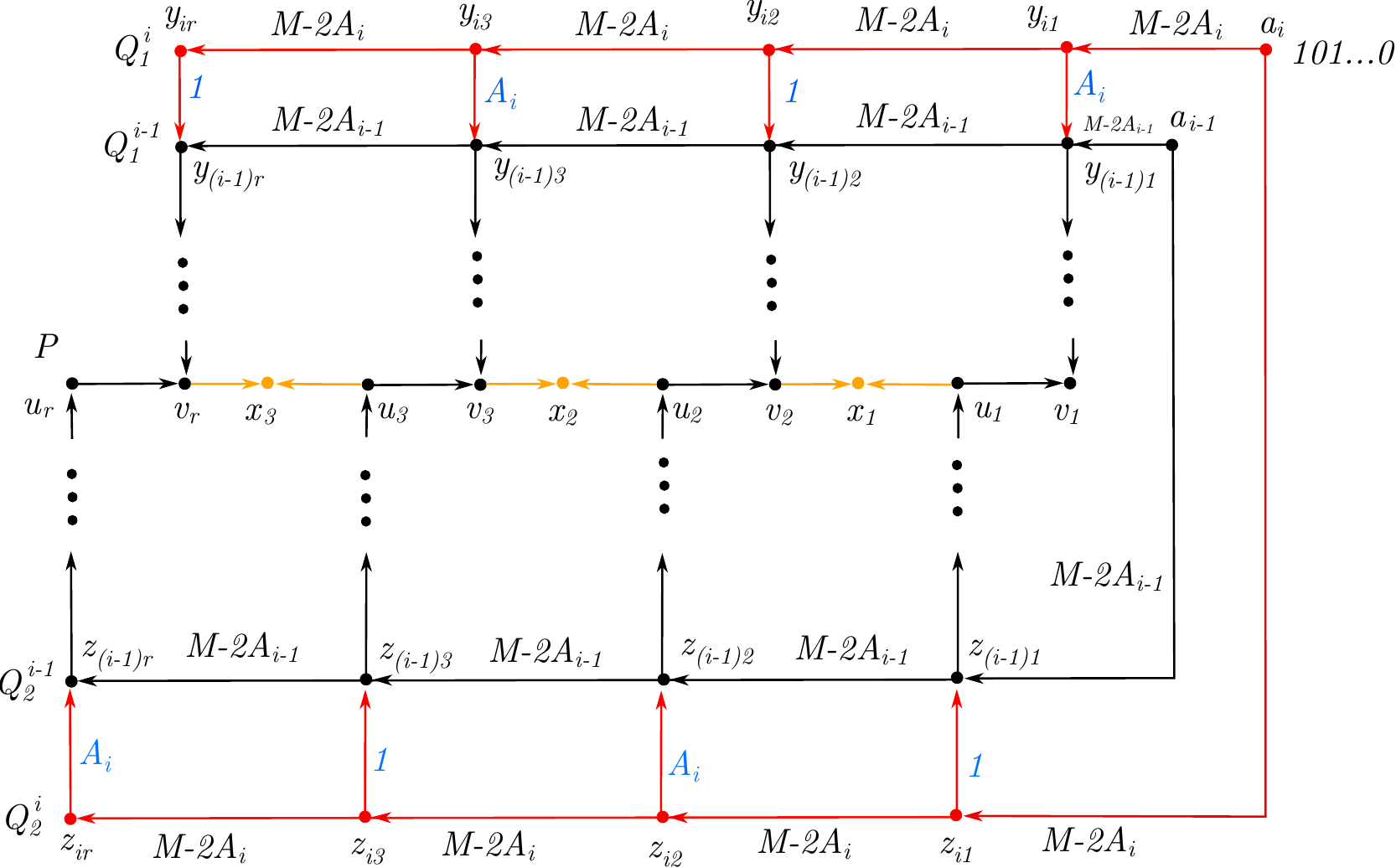}}
	\caption{Adding $a_i$ to the current graph $G_{i-1}$.  
	}\label{fig:Addvertex}
\end{figure}

The path $P$ will separate the plane into two parts: the upper part (or $0$ part) and the lower part (or $1$ part); see \Cref{fig:pathP}(b). The upper part will realize the fact that some edges of $X$ are NOT added to the shortest path tree  of $a_i$ and the lower part serves the opposite purpose. In particular, if the $j$-th bit of $s_i$ is 0, then the shortest path from $a_i$ to $v_j$ will only contain edges from the upper part and hence does not contain $e_j$. Otherwise, the shortest path from $a_i$ to $v_j$ will only contain edges from the lower part and $e_j$. The rule for adding $a_i$ is as follows (see \Cref{fig:Addvertex}):
\begin{enumerate}
	\item Add two directed paths $Q^i_1 = (a_i = y_{i0} \rightarrow y_{i1} \ldots \rightarrow y_{ir})$ and $Q^i_2 = (a_i = z_{i0} \rightarrow z_{i1} \ldots \rightarrow z_{ir})$ directed away from $a_i$. Each path has exactly $r$ edges. Each edge of the two paths is assigned a weight $M - 2\cdot A_i$. $Q^i_1$ is embedded in the upper part of the plane, separated by $P$, and $Q^i_2$ is embedded in the lower part of the plane.
	
	\item Look at the bit string $s_i$ (see an example in \Cref{fig:Addvertex}, the first 3 bits are $\{1,0,1\}$, and the last bit is $0$ in the string of $a_i$). Assume that $i\geq 1$.
	\begin{itemize}

		\item if the $j$-th bit is 1, add a directed edge $(y_{ij} \rightarrow y_{(i-1)j})$ of weight $A_i$, and add a directed edge $(z_{ij}\rightarrow z_{(i-1)j})$ of weight $1$.  Note that $y_{(i-1)j}$ and $z_{(i-1)j}$  are vertices on the paths $Q^{i-1}_{1}$ and  $Q^{i-1}_{2}$, respectively, of $a_{i-1}$.  For example, in \Cref{fig:Addvertex}, the first bit of $s_i$ is $1$ and hence we have an edge of weight $A_i$ from $y_{i1}$ to $y_{(i-1)1}$ and  an edge of weight $1$ from $z_{i1}$ to $z_{(i-1)1}$. The same holds for $y_{i3}$ and $z_{i3}$. 
		
		\item if the $j$-th bit is 0, add a directed edge $(y_{ij} \rightarrow y_{(i-1)j})$ of weight $1$, and add a directed edge $(z_{ij}\rightarrow z_{(i-1)j})$ of weight $A_i$.  For example, in \Cref{fig:Addvertex}, the second bit of $s_i$ is $0$ and hence we have an edge of weight $1$ from $y_{i2}$ to $y_{(i-1)2}$ and  an edge of weight $A_i$ from $z_{i2}$ to $z_{(i-1)2}$. The same holds for $y_{ir}$ and $z_{ir}$.  	
	\end{itemize}
	When $i = 0$ (see \Cref{fig:pathP}(b)), we connect $y_{0j}$ to $v_j$ and $z_{0j}$ to $u_j$ for every $j\in [1,r]$. Note that, since $s_0$ only contains 0 bits, every edge $(y_{0j}\rightarrow v_j)$ has weight 1 and every edge $(z_{0j}\rightarrow u_j)$ has weight $A_0$. In this case,  no edge in $X$ will be included in the shortest path tree of $a_0$.  
	
\end{enumerate}

\begin{figure}[!htb]
	\centering{\includegraphics[width=0.9\textwidth]{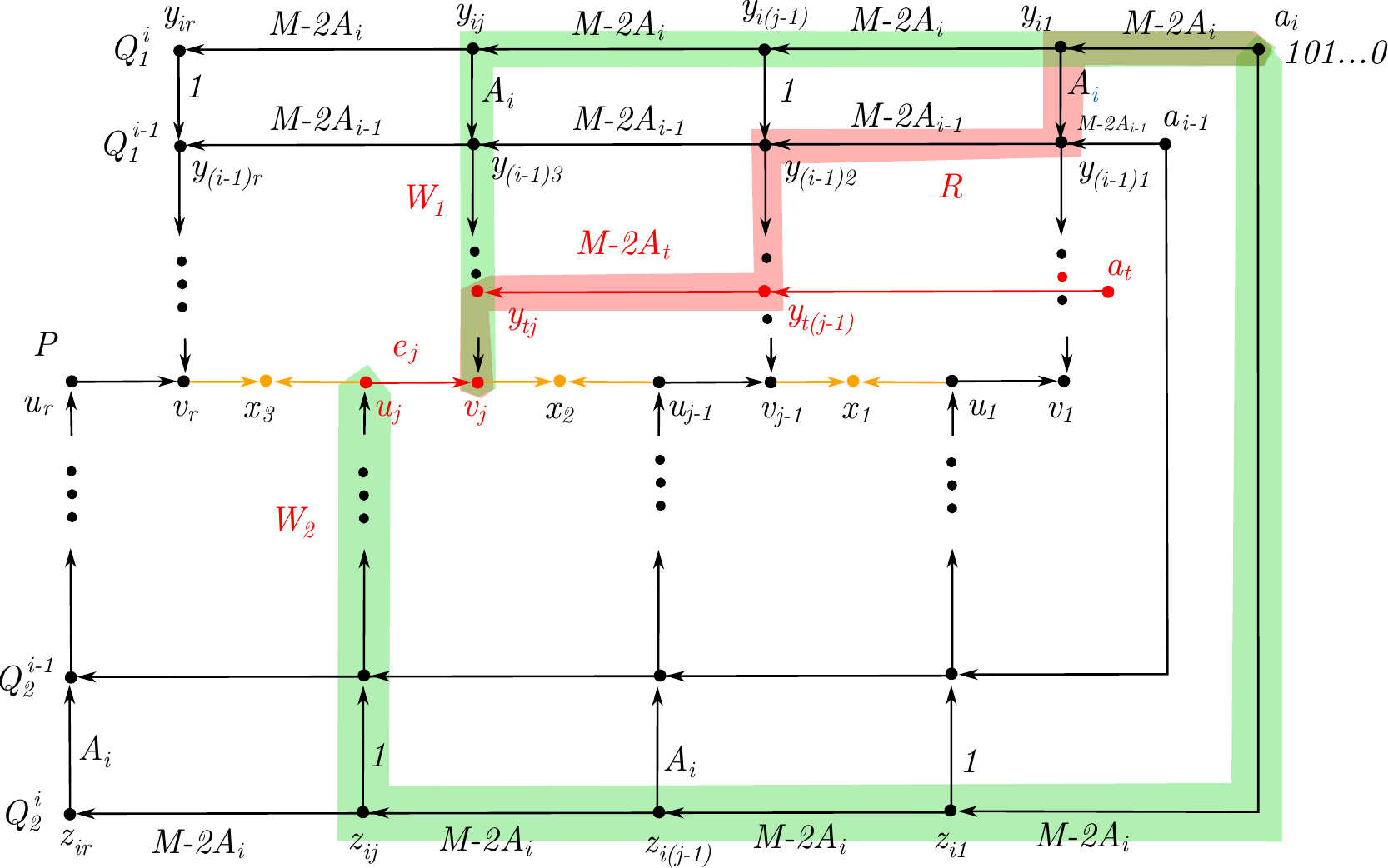}}
	\caption{Two green-highlighted paths are shortest path from $a_i$ to $v_j$ and to $u_j$ in graphs $G\setminus \{e_j\}$ and $G_i$, respectively.  
	}\label{fig:PathAnalysis}
\end{figure}

We now analyze the shortest path tree of $a_i$, denote by $T_i$.  We say that an edge is \emph{horizontal} if it belongs to $P$ or $Q^t_1$ or $Q^t_2$ for some $t\in [0,2^r-1]$; otherwise, we say that the edge is \emph{vertical}. Note that a vertical edge is either an edge from a vertex $y_{tj}\in Q^{t}_1$ \emph{down} to some vertex in $Q^{t-1}_1$ when $t\geq 1$ or down to some vertex in $P$ when $t=0$, or from a vertex $z_{tj} \in Q^t_2$ \emph{up} to a vertex in $ Q^{t-1}_2$ when $t\geq 1$ or to a vertex in $P$ when $t=0$.

The following claim is the key to the proof; see \Cref{fig:PathAnalysis}.

\begin{claim}\label{clm:ai-to-vj} For any $j\in [1,r]$, then the shortest path from $a_i$ to $v_j$ \emph{in $G\setminus \{e_j\}$} consists of: (a) the subpath from $a_i$ to $y_{ij}$ of $Q^i_1$, which only contains horizontal edges,  and the directed paths from $y_{ij}$ to $v_j$, which only contains vertical edges. 
\end{claim} 
\begin{proof}
	The shortest path from $a_i$ to $v_j$ is the path highlighted green in \Cref{fig:PathAnalysis}. Let $W_1$ be the path from $a_i$ to $v_j$ as described in the claim. Note that there is a unique directed path from $y_{ij}$ to $v_j$, which only contains vertical edges.  Similarly, there is a unique directed path from $a_i$ to $y_{ij}$ which only contains horizontal edges.  Thus, if the shortest path from $a_i$ to $v_j$ goes through $y_{ij}$, the path must be $W_1$. 
	
	Let $R$ be the shortest path from $a_i$ to $v_j$; assume that $R\not= W_1$. Thus, $y_{ij}\not\in R$.  Observe that $R$ can only include edges in the upper part of $P$ since $e_j$ is deleted from $G$. That is, $R$ does not contain any $z$-vertex. Furthermore, $R$ contains exactly $j$ horizontal edges and $i$ vertical edges. Note that $W_1$ also contains exactly $j$ horizontal edges and $i$ vertical edges. Additionally, the weight any horizontal edge of $W_1$ is at most the weight of  any horizontal edge of $R$ by the choice of $\{A_i\}_{i=0}^{2^{r}-1}$.
	
	Let $(y_{t(j-1)}\rightarrow y_{tj})$ be the last horizontal edge of $R$; see \Cref{fig:PathAnalysis}. As $y_{ij}\not\in R$, $t < i$. Thus we have:
	\begin{equation}
		\begin{split}
			w(R) - w(W_1) &\geq w(y_{t(j-1)}\rightarrow y_{tj}) - w(y_{i(j-1)}\rightarrow y_{ij}) - \mbox{( total weight of all vertical edges of $W_1$)}\\
			&\geq 	M - 2A_t - (M-2A_i) - (A_i + A_{i-1} + \ldots A_0)\\
			&\geq A_i - 2A_t - (A_{i-1} + \ldots A_0)\geq 1 \quad \mbox{(by \Cref{eq:Ai})}
		\end{split}
	\end{equation}
	This contradicts that $R$ is a shortest path from $a_i$ to $v_j$.
\end{proof}

The following claim is similar to \Cref{clm:ai-to-vj}, except that the graph is $G$. The proof is exactly the same.

\begin{claim}\label{clm:ai-to-uj} For any $j\in [1,r]$, the shortest path from $a_i$ to $u_j$ in $G$ consists of: (a) the subpath from $a_i$ to $z_{ij}$ of $Q^i_2$, which only contains horizontal edges,  and the directed paths from $z_{ij}$ to $u_j$, which only contains vertical edges. 
\end{claim} 

Lastly, we claim that $e_j$ is in $T_i$ or not is equivalent to whether the $j$-th bit in the bit string $s_i$ is $1$ or $0$. This implies that the directed edge set system does not have a bounded VC dimension.

\begin{claim}\label{clm:ej-vs-si} For any $j\in [1,r]$, if $s_i[j]=1$ then $e_j \in T_i$; otherwise, $e_j\not\in T_i$
\end{claim} 
\begin{proof}
	Let $W_1$ be the shortest path from $a_i$ to $v$ in $G\setminus e_j$, and $W_2$ be the shortest path from $a_i$ to $u_j$ in $G$.  Suppose that $s_i[j] = 1$, then  $w(y_{ij}\rightarrow y_{(i-1)j}) = A_i$ and 	$w(z_{ij}\rightarrow z_{(i-1)j}) = 1$. See \Cref{fig:PathAnalysis}. Note that both $W_1$ and $W_2$ contain the same number of horizontal edges, each of the same weight $M - 2A_i$. Thus, we have:
	\begin{equation}
		\begin{split}
			w(W_1) - w(W_2\circ e_j) &\geq A_i -  \mbox{( total weight of all vertical edges of $W_2$)} - 1\\
			&\geq A_i - (1 + A_{i-1} + \ldots A_0) - 1 > 0  \quad \mbox{(by \Cref{eq:Ai})}
		\end{split}
	\end{equation}
	Thus, $W_2\circ e_j$ is the shortest path from $a_i$ to $v_j$, implying the claim. The proof that if  $s_i[j]=0$ then $e_j\not \in T_i$ follows the same line.
\end{proof}

%% file: Main.bbl
\newcommand{\etalchar}[1]{$^{#1}$}
\begin{thebibliography}{GMWWN18}

\bibitem[AWW16]{AWW16}
A.~Abboud, V.~V. Williams, and J.~Wang.
\newblock Approximation and fixed parameter subquadratic algorithms for radius
  and diameter in sparse graphs.
\newblock In {\em Proceedings of the 27th Annual ACM-SIAM Symposium on Discrete
  Algorithms}, SODA '16, page 377–391, 2016.

\bibitem[Bak94]{Baker94}
B.~S. Baker.
\newblock Approximation algorithms for {NP}-complete problems on planar graphs.
\newblock {\em Journal of the ACM}, 41(1):153--180, 1994.

\bibitem[BBE{\etalchar{+}}21]{BCEJLMP21}
Ni. Bousquet, W.~Cames~Van Batenburg, L.~Esperet, G.~Joret, W.~Lochet,
  C.~Muller, and F.~Pirot.
\newblock Packing and covering balls in graphs excluding a minor.
\newblock {\em Combinatorica}, 41(3):299--318, 2021.

\bibitem[BC14]{BC14}
G.~Borradaile and E.~W. Chambers.
\newblock Covering nearly surface-embedded graphs with a fixed number of balls.
\newblock {\em Discrete \& Computational Geometry}, 51(4):979--996, 2014.

\bibitem[BP21]{BP21}
G.~Bodwin and M.~Parter.
\newblock Restorable shortest path tiebreaking for edge-faulty graphs.
\newblock In {\em Proceedings of the 2021 {ACM} Symposium on Principles of
  Distributed Computing}, 2021.

\bibitem[BT15]{BT15}
N.~Bousquet and S.~Thomass{\'{e}}.
\newblock {VC}-dimension and {E}rd{\H{o}}s{\textendash}{P}{\'{o}}sa property.
\newblock {\em Discrete Mathematics}, 338(12):2302--2317, 2015.

\bibitem[Cab18]{Cabello18}
S.~Cabello.
\newblock Subquadratic algorithms for the diameter and the sum of pairwise
  distances in planar graphs.
\newblock {\em ACM Transactions on Algorithms}, 15(2), 2018.
\newblock Announced at SODA'17.

\bibitem[CADWN17]{CDW17}
V.~Cohen-Addad, S.~Dahlgaard, and C.~Wulff-Nilsen.
\newblock Fast and compact exact distance oracle for planar graphs.
\newblock In {\em {IEEE} 58th Annual Symposium on Foundations of Computer
  Science}, FOCS `17, pages 962--973, 2017.

\bibitem[CC07]{CC07}
Sergio Cabello and Erin~W. Chambers.
\newblock Multiple source shortest paths in a genus $g$ graph.
\newblock In {\em Proceedings of the 18th Annual ACM-SIAM Symposium on Discrete
  Algorithms}, SODA '07, page 89–97. Society for Industrial and Applied
  Mathematics, 2007.

\bibitem[CCE13]{CCE13}
S.~Cabello, E.~W. Chambers, and J.~Erickson.
\newblock Multiple-source shortest paths in embedded graphs.
\newblock {\em {SIAM} Journal on Computing}, 42(4):1542--1571, 2013.

\bibitem[CEV07]{CEV07}
V.~Chepoi, B.~Estellon, and Y.~Vaxes.
\newblock Covering planar graphs with a fixed number of balls.
\newblock {\em Discrete $\&$ Computational Geometry}, 37(2):237--244, 2007.

\bibitem[CGMW19]{CGMW19}
P.~Charalampopoulos, P.~Gawrychowski, S.~Mozes, and O.~Weimann.
\newblock Almost optimal distance oracles for planar graphs.
\newblock In {\em Proceedings of the 51st Annual {ACM} {SIGACT} Symposium on
  Theory of Computing}, STOC `19, pages 138--151, 2019.

\bibitem[CK97]{CK97}
V.~Chepoi and S.~Klav{\v{z}}ar.
\newblock The wiener index and the szeged index of benzenoid systems in linear
  time.
\newblock {\em Journal of Chemical Information and Computer Sciences},
  37(4):752--755, 1997.

\bibitem[CK09]{CK09}
S.~Cabello and C.~Knauer.
\newblock Algorithms for graphs of bounded treewidth via orthogonal range
  searching.
\newblock {\em Computational Geometry}, 42(9):815--824, 2009.

\bibitem[DFHT05]{DFHT5}
Erik~D. Demaine, Fedor~V. Fomin, Mohammadtaghi Hajiaghayi, and Dimitrios~M.
  Thilikos.
\newblock Subexponential parameterized algorithms on bounded-genus graphs and
  $h$-minor-free graphs.
\newblock {\em Journal of the {ACM}}, 52(6):866--893, 2005.

\bibitem[DH05]{DH05}
E.~D. Demaine and M.~Hajiaghayi.
\newblock Bidimensionality: New connections between { FPT} algorithms and
  {PTAS}s.
\newblock In {\em Proceedings of the Sixteenth Annual ACM-SIAM Symposium on
  Discrete Algorithms}, SODA'05, pages 590--601, 2005.

\bibitem[DHM10]{DMM10}
Erik~D. Demaine, MohammadTaghi Hajiaghayi, and Bojan Mohar.
\newblock Approximation algorithms via contraction decomposition.
\newblock {\em Combinatorica}, 30(5):533--552, 2010.

\bibitem[DHT04]{DHT04}
Erik~D. Demaine, MohammadTaghi Hajiaghayi, and Dimitrios~M. Thilikos.
\newblock The bidimensional theory of bounded-genus graphs.
\newblock In {\em Proceedings of the 29th Symposium on Mathematical Foundations
  of Computer Science}, MFCS '04, pages 191--203. 2004.

\bibitem[DHV20]{DHV20}
D.~Ducoffe, M.~Habib, and L.~Viennot.
\newblock Diameter computation on h-minor free graphs and graphs of bounded
  (distance) vc-dimension.
\newblock In {\em Proceedings of the 31st Annual ACM-SIAM Symposium on Discrete
  Algorithms}, SODA '20, page 1905–1922, 2020.

\bibitem[DPBF09]{DPBF9}
Frederic Dorn, Eelko Penninkx, Hans~L. Bodlaender, and Fedor~V. Fomin.
\newblock Efficient exact algorithms on planar graphs: Exploiting sphere cut
  decompositions.
\newblock {\em Algorithmica}, 58(3):790--810, 2009.

\bibitem[Epp03]{Eppstein03}
D.~Eppstein.
\newblock Dynamic generators of topologically embedded graphs.
\newblock In {\em Proceedings of the 14th Annual ACM-SIAM Symposium on Discrete
  Algorithms}, SODA `03, pages 599–--608, 2003.

\bibitem[Eri10]{Erickson10}
J.~Erickson.
\newblock Maximum flows and parametric shortest paths in planar graphs.
\newblock In {\em Proceedings of the 21st Annual {ACM}-{SIAM} Symposium on
  Discrete Algorithms}, 2010.

\bibitem[Fed87]{Federickson87}
Greg~N. Federickson.
\newblock Fast algorithms for shortest paths in planar graphs, with
  applications.
\newblock {\em {SIAM} Journal on Computing}, 16(6):1004--1022, 1987.

\bibitem[FHMWN20]{FMW20}
V.~Fredslund-Hansen, S.~Mozes, and C.~Wulff-Nilsen.
\newblock Truly subquadratic exact distance oracles with constant query time
  for planar graphs.
\newblock {\em arXiv preprint arXiv:2009.14716}, 2020.
\newblock \url{https://arxiv.org/abs/2009.14716}.

\bibitem[FR01]{FR01}
J.~Fakcharoenphol and S.~Rao.
\newblock Planar graphs, negative weight edges, shortest paths, and near linear
  time.
\newblock In {\em Proceedings 42nd {IEEE} Symposium on Foundations of Computer
  Science}, FOCS `01, 2001.

\bibitem[GKM{\etalchar{+}}21]{GKMSW21}
P.~Gawrychowski, H.~Kaplan, S.~Mozes, M.~Sharir, and O.~Weimann.
\newblock Voronoi diagrams on planar graphs, and computing the diameter in
  deterministic $\tilde{O}(n^{5/3})$ time.
\newblock {\em {SIAM} Journal on Computing}, (2):509--554, 2021.

\bibitem[GMWWN18]{GMWW18}
P.~Gawrychowski, S.~Mozes, O.~Weimann, and C.~Wulff-Nilsen.
\newblock Better tradeoffs for exact distance oracles in planar graphs.
\newblock In {\em Proceedings of the 29th Annual {ACM}-{SIAM} Symposium on
  Discrete Algorithms}, number SODA `18, pages 515--529, 2018.

\bibitem[Hus17]{Husfeldt17}
T.~Husfeldt.
\newblock {Computing Graph Distances Parameterized by Treewidth and Diameter}.
\newblock In {\em 11th International Symposium on Parameterized and Exact
  Computation (IPEC 2016)}, volume~63, pages 16:1--16:11, 2017.

\bibitem[HW87]{HW87}
D.~Haussler and E.~Welzl.
\newblock {$\varepsilon$}-nets and simplex range queries.
\newblock {\em Discrete {\&} Computational Geometry}, 2(2):127--151, 1987.

\bibitem[JR23]{JR23}
G.~Joret and C.~Rambaud.
\newblock Neighborhood complexity of planar graphs.
\newblock {\em arXiv preprint arXiv:2302.12633}, 2023.

\bibitem[KKR{\etalchar{+}}97]{KKRUW97}
E.~Kranakis, D.~Krizanc, B.~Ruf, J.~Urrutia, and G.~Woeginger.
\newblock The {VC}-dimension of set systems defined by graphs.
\newblock {\em Discrete Applied Mathematics}, 77(3):237--257, 1997.

\bibitem[KKR12]{KKR12}
K.~Kawarabayashi, Y.~Kobayashi, and B.~Reed.
\newblock The disjoint paths problem in quadratic time.
\newblock {\em Journal of Combinatorial Theory, Series B}, 102(2):424--435,
  2012.

\bibitem[Kle05a]{Klein05}
P.~N. Klein.
\newblock A linear-time approximation scheme for planar weighted {TSP}.
\newblock In {\em Proceedings of the 46th Annual IEEE Symposium on Foundations
  of Computer Science}, FOCS '05, pages 647--657, 2005.

\bibitem[Kle05b]{Klein05B}
Philip.~N. Klein.
\newblock Multiple-source shortest paths in planar graphs.
\newblock In {\em Proceedings of the 16th Annual ACM-SIAM Symposium on Discrete
  Algorithms}, SODA '05, page 146–155. Society for Industrial and Applied
  Mathematics, 2005.

\bibitem[KR10]{KR10}
K.~Kawarabayashi and B.~Reed.
\newblock A separator theorem in minor-closed classes.
\newblock In {\em Proceedings of the 51st Annual Symposium on Foundations of
  Computer Science}, FOCS '10, 2010.

\bibitem[KTW18]{KTW18}
Ke. Kawarabayashi, R.~Thomas, and P.~Wollan.
\newblock A new proof of the flat wall theorem.
\newblock {\em Journal of Combinatorial Theory, Series B}, 129:204--238, 2018.

\bibitem[KTW20]{KTW20}
K.~Kawarabayashi, R.~Thomas, and P.~Wollan.
\newblock Quickly excluding a non-planar graph, 2020.

\bibitem[Le23]{Le22}
H.~Le.
\newblock Approximate distance oracles for planar graphs with subpolynomial
  error dependency.
\newblock In {\em Proceedings of the 18th Annual ACM-SIAM Symposium on Discrete
  Algorithms}, SODA '23, pages 1877--1904. Society for Industrial and Applied
  Mathematics, 2023.

\bibitem[LP19]{LP19}
J.~Li and M.~Parter.
\newblock Planar diameter via metric compression.
\newblock In {\em Proceedings of the 51st Annual ACM SIGACT Symposium on Theory
  of Computing}, STOC 2019, page 152–163, 2019.

\bibitem[LP21]{LP20}
Y.~Long and S.~Pettie.
\newblock Planar distance oracles with better time-space tradeoffs.
\newblock In {\em Proceedings of the 2021 ACM-SIAM Symposium on Discrete
  Algorithms, SODA'21}, pages 2517--2537, 2021.

\bibitem[LT79]{LT79}
R.~Lipton and R.~Tarjan.
\newblock A separator theorem for planar graphs.
\newblock {\em SIAM Journal on Applied Mathematics}, 36(2):177--189, 1979.

\bibitem[LT80]{LT80}
R.~J. Lipton and R.~E. Tarjan.
\newblock Applications of a planar separator theorem.
\newblock {\em SIAM Journal on Computing}, 9(3):615--627, 1980.

\bibitem[MNNW18]{MNNW18}
Shay Mozes, Kirill Nikolaev, Yahav Nussbaum, and Oren Weimann.
\newblock Minimum cut of directed planar graphs in $o(n\log\log n)$ time.
\newblock In {\em Proceedings of the 29th Annual {ACM}-{SIAM} Symposium on
  Discrete Algorithms}, pages 477--494. 2018.

\bibitem[MS12]{MS12}
S.~Mozes and C.~Sommer.
\newblock Exact distance oracles for planar graphs.
\newblock In {\em Proceedings of the 23rd Annual {ACM}-{SIAM} Symposium on
  Discrete Algorithms}, SODA`12, pages ﻿209--222, 2012.

\bibitem[RS83]{RS83}
N.~Robertson and P.~D. Seymour.
\newblock Graph minors. {I.} {E}xcluding a forest.
\newblock {\em Journal of Combinatorial Theory, Series B}, 35(1):39--61, 1983.

\bibitem[RS03]{RS03}
N.~Robertson and P.~D. Seymour.
\newblock Graph minors. {XVI}. {E}xcluding a non-planar graph.
\newblock {\em Journal of Combinatoral Theory Series B}, 89(1):43--76, 2003.

\bibitem[RS04]{RS04}
N.~Robertson and P.~D. Seymour.
\newblock Graph minors. {XX}. {W}agner's conjecture.
\newblock {\em Journal of Combinatorial Theory Series B}, 92(2):325--357, 2004.

\bibitem[RW09]{WR09}
Bruce Reed and David~R. Wood.
\newblock A linear-time algorithm to find a separator in a graph excluding a
  minor.
\newblock {\em {ACM} Transactions on Algorithms}, 5(4):1--16, 2009.

\bibitem[Tho04]{Thorup04}
M.~Thorup.
\newblock Compact oracles for reachability and approximate distances in planar
  digraphs.
\newblock {\em Journal of the ACM}, 51(6):993–1024, 2004.
\newblock Announced at FOCS' 01.

\bibitem[VC71]{VC71}
V.~N. Vapnik and A.~Ya. Chervonenkis.
\newblock On the uniform convergence of relative frequencies of events to their
  probabilities.
\newblock {\em Theory of Probability \& Its Applications}, 16(2):264--280,
  1971.

\bibitem[VV86]{VV86}
L.G. Valiant and V.V. Vazirani.
\newblock {NP} is as easy as detecting unique solutions.
\newblock {\em Theoretical Computer Science}, 47(0):85 -- 93, 1986.

\bibitem[Wie47]{Wiener47}
H.~Wiener.
\newblock Structural determination of paraffin boiling points.
\newblock {\em Journal of the American Chemical Society}, 69(1):17--20, 1947.

\bibitem[WN09]{WulffNilsen2009}
C.~Wulff-Nilsen.
\newblock Wiener index and diameter of a planar graph in subquadratic time.
\newblock In {\em Proceedings of the 25th European Workshop on Computational
  Geometry}, pages 25--28, 2009.

\bibitem[WN11]{WulffNilsen11}
C.~Wulff-Nilsen.
\newblock Separator theorems for minor-free and shallow minor-free graphs with
  applications.
\newblock In {\em Proceedings of the 52nd Annual Symposium on Foundations of
  Computer Science}, FOCS '11, 2011.

\bibitem[WN14]{WulffNilsen14}
C.~Wulff-Nilsen.
\newblock Faster separators for shallow minor-free graphs via dynamic
  approximate distance oracles.
\newblock In {\em Proceedings of the 41th International Colloquium on Automata,
  Languages, and Programming}, ICALP '14, pages 1063--1074. 2014.

\end{thebibliography}
